\documentclass[journal]{IEEEtran}
\usepackage[utf8]{inputenc}
\usepackage{cite}
\usepackage{dsfont}
\usepackage{bm}
\usepackage{comment}
\usepackage{amsmath,amssymb,amsfonts,amsthm}
\usepackage{array}
\usepackage[ruled,vlined,linesnumbered]{algorithm2e}
\usepackage{algorithmic}
\usepackage[caption=false]{subfig}
\usepackage[usenames,dvipsnames]{xcolor}
\usepackage{graphicx}
\usepackage[english]{babel}
\newtheorem{theorem}{Theorem}
\newtheorem{lemma}{Lemma}
\newtheorem{example}{Example}
\newtheorem{proposition}{Proposition}
\newtheorem{corollary}{Corollary}
\newtheorem{definition}{Definition}

\newcommand{\Rmnum}[1]{\uppercase\expandafter{\romannumeral #1}}

\ifodd 1
    \newcommand{\rev}[1]{{\color{blue}#1}}
    
    \newcommand\outline[1]{{\color{RoyalBlue}#1}}
    \newcommand{\com}[2]{\textbf{\color{blue} (COMMENT from [#1]: #2)}}
\else
    \newcommand{\rev[1]}{{#1}}
    
    \newcommand{\com}[2]{}
\fi

\makeatletter
\newcommand*\bigcdot{\mathpalette\bigcdot@{.5}}
\newcommand*\bigcdot@[2]{\mathbin{\vcenter{\hbox{\scalebox{#2}{$\m@th#1\bullet$}}}}}
\makeatother

\ifCLASSINFOpdf
\else
\fi
\begin{document}
%
\title{Optimal~$(0,1)$-Matrix Completion \\with Majorization Ordered Objectives}
%
%
%

\author{Yanfang~Mo,~\IEEEmembership{Member,~IEEE,}
        Wei~Chen,~\IEEEmembership{Member,~IEEE,}
        Keyou~You,~\IEEEmembership{Senior~Member,~IEEE}
        and~Li~Qiu,~\IEEEmembership{Fellow,~IEEE\\ \textbf{To the memory of Pravin Varaiya}}
\thanks{This work was partially supported by the Shenzhen-Hong Kong-Macau
Science and Technology Innovation Fund under project~(SZSTI21EG08), the Research Grants Council of Hong Kong, China, under the Theme-Based Research Scheme~(T23-701/14-N), Schneider Electric, Lenovo Group~(China) Limited, and the Hong Kong Innovation and Technology Fund~(ITS/066/17FP) under the HKUST-MIT Research Alliance Consortium. } 
\thanks{Y. Mo is with the School of Data Science, City University of Hong Kong, Tat Chee Avenue, Kowloon, Hong Kong, China (e-mail: yanfang.mo@cityu.edu.hk, ymoaa@connect.ust.hk).}
\thanks{W. Chen is with the Department of Mechanics and Engineering Science \& State Key Laboratory for Turbulence and Complex Systems, Peking University, Beijing 100871, China (e-mail: w.chen@pku.edu.cn).}
\thanks{K. You is with the Department of Automation, and Beijing National Research Center for Information Science and Technology, Tsinghua University, Beijing 100084, China (e-mail: youky@tsinghua.edu.cn).}
\thanks{L. Qiu is with the Department of Electronic and Computer Engineering, The Hong Kong University of Science and Technology, Clear Water Bay, Kowloon, Hong Kong, China (e-mail: eeqiu@ust.hk).}
}

\maketitle

\begin{abstract}
We propose and examine two optimal~$(0,1)$-matrix completion problems with majorization ordered objectives. They elevate the seminal study by Gale and Ryser from feasibility to optimality in partial order programming (POP), referring to optimization with partially ordered objectives. We showcase their applications in electric vehicle charging, portfolio optimization, and secure data storage. Solving such integer POP (iPOP) problems is challenging because of the possible non-comparability among objective values and the integer requirements. Nevertheless, we prove the essential uniqueness of all optimal objective values and identify two particular ones for each of the two inherently symmetric iPOP problems. Furthermore, for every optimal objective value, we decompose the construction of an associated optimal~$(0,1)$-matrix into a series of sorting processes, respectively agreeing with the rule of thumb ``peak shaving'' or ``valley filling.'' We show that the resulting algorithms have linear time complexities and verify their empirical efficiency via numerical simulations compared to the standard order-preserving method for POP.
\end{abstract}


%
\IEEEpeerreviewmaketitle

\section{Introduction}\label{sec_introduction}

Zero-one matrix completion plays a prominent role in many areas like network construction~\cite{dau2014simple}, experimental block design~\cite{giovagnoli1981optimum}, bidimensional election~\cite{lari2014bidimensional}, discrete tomography~\cite{batenburg2005evolutionary}, and {indivisible resource allocation~\cite{mo2022Tensor}.} This line of research originates from the class of~$(0,1)$-matrices with specified row and column sums, namely
\begin{equation} \label{matrixclass}
\left\{[a_{ij}]\in\{0,1\}^{m\times n}~\middle\vert~ \sum_{j=1}^{n}a_{ij}=r_i,  \sum_{i=1}^{m}a_{ij}=c_j\right\},
\end{equation}
where~$r_i$ and~$c_j$ are respectively the prescribed line sums of the~$i$th row and~$j$th column for the~$m$ rows and~$n$ columns. Gale~\cite{gale1957theorem} and Ryser~\cite{ryser1957combinatorial} independently derived an inequality under the majorization order to characterize the existence of a matrix satisfying~\eqref{matrixclass}. After that, there have been a number of extensions in the literature~\cite{mo2022Tensor,nelson20150,mo2018staircase,fulkerson1959network,brualdi2019gale}. Most of them focus on completing matrices that belong to a subset of the mentioned matrix class, e.g., prescribing zeros/ones in certain positions~\cite{brualdi2003matrices,chen2016constrained,2007BrualdiDahl,mo2020market,anstee1982properties,fulkerson1960zero} and imposing structural constraints~\cite{chen1966realization,erdosgallai1960degree,berger2011dag}. See a more detailed discussion in Section~\ref{sec_relatedwork}.


However, in practice, we often wonder which completion is better.  
This fact motivates us to take the study of~$(0,1)$-matrix completion from feasibility to optimality. In this work, {we aim to find a~$(0,1)$-matrix~$[a_{ij}]\in\{0,1\}^{m\times n}$ to optimize the flatness of~$$\bm d\mp \left[\sum_{i=1}^{m}a_{i1}~\sum_{i=1}^{m}a_{i2}~\cdots~\sum_{i=1}^{m}a_{in}\right]',$$ subject to~$\sum_{j=1}^{n}a_{ij}=r_i$, where~$\bm r\in \mathbb{R}^m$ prescribes the row sums and~$\bm d\in \mathbb{R}^n$ is a given reference vector to the column sums of the optimized~$(0,1)$-matrix.}

The flatness objective comes from electric vehicle (EV) charging, where we require the smoothest remaining supply or aggregated load profile to ease the supply/demand balance. Other applications include portfolio optimization~\cite{markowitz2010portfolio} and secure data storage~\cite{robling1982cryptography}, where the flatness of the objective vector reflects the risk level of a portfolio of assets or the security level of a distributed data storage service.

{Mathematically, we can apply the majorization order to measure the flatness of a vector~\cite{yue1973optimality,olkin2016inequalities,capponi2016liability}.}
This practice is because majorization describes how evenly a quantity is dispersed in a vector and has been popular in measuring statistical dispersion~\cite{chang1993rearrangement} or economic disparity~\cite{mosler1994majorization}. Meanwhile, the majorization inequality in the Gale-Ryser theorem motivates us to study optimal~$(0,1)$-matrix completion with majorization ordered objectives~\cite{gale1957theorem,ryser1957combinatorial}. Then, we complement the existing studies by using majorization to evaluate the performance of matrix completion (optimality), in addition to characterizing the existence of particular matrices (feasibility).

To the best of our knowledge, we are the first to study optimal~$(0,1)$-matrix completion problems with vector-valued objectives ordered by majorization. The majorization ordered objectives bring unique challenges. Like other partial order programming (POP) problems, the optimization under majorization suffers from the question of whether the optimal objective value is unique or not because the attainable objectives are partially ordered rather than totally ordered, and two objective values may not be comparable. Even worse, the integer constraints lead to integer POP (iPOP) and further complicate the optimal matrix completion. Overall, it is challenging to find one of the optimal solutions to an iPOP problem, not to mention all of them. Nevertheless, we propose and address two iPOP problems under majorization, elevating the study in~\cite{gale1957theorem,ryser1957combinatorial} and generalizing several results in \cite{kim1998simple,wan2000global,mo2022Tensor}.

We show that all the optimal objective values for each proposed iPOP problem are essentially unique in the sense that they are rearrangements of each other. Although not every rearrangement of an optimal objective value is attainable, we identify two particular ones characterized by the order of elements in the objective value or the corresponding column sum vector. After that, we propose a ``peak-shaving'' or ``valley-filling'' approach to every optimal objective value for each iPOP problem. Notably, the resulting algorithms decompose the construction of an associated optimal~$(0,1)$-matrix into a series of sorting processes. Specifically, they sequentially construct the rows of the optimal~$(0,1)$-matrix respectively by decreasing the largest elements or increasing the smallest ones of a vector, agreeing with the rule of thumb ``peak shaving'' or ``valley filling.'' This fact provides a fine-grained perspective on the inherent symmetry of the two iPOP problems. More importantly, the algorithms can avoid round-off errors, have linear time complexities, and are thus more useful in large-scale matrix completion instances with many rows/columns than conceivable alternatives, as substantiated by simulations. 

From a historical perspective, our approach inherits the merits of the Ryser-like algorithms, which are generally designed to construct~$(0,1)$-matrices with given row/column sums and can help check their existence  numerically~\cite[Section~3]{brualdi2006algorithms,brualdi2006combinatorial}. From this point of view, our partially ordered objectives and algorithms jointly answer the intriguing question: which~$(0,1)$-matrix does a Ryser-like algorithm return? This observation coincides with the study of another interesting iPOP involving the Nobel prize-winning work on stable matchings; see Section~\ref{sec_relatedwork} for more details. 

Branch and bound~(B\&B) and scalarization are two common methods for iPOP. The former usually leads to time-consuming algorithms in worst cases~\cite{morrison2016branch,norkin2019b,przybylski2017multi}. Besides, it is challenging to develop branching strategies and find the bounds for a set of elements under majorization; thus, the direct use of B\&B cannot work for our iPOP problems. The latter features converting the partially ordered vector-valued objectives into real-valued ones by order-preserving functions (referring to Schur-convex functions for majorization), ignoring the partially ordered structures of feasible solutions and objective values. Notably, we show that the integer programs after scalarization are tractable due to their separable convex objective functions and totally unimodular constraint matrices. Nevertheless, we still need a computation-consuming process to approach an exact integer solution~\cite{meyer1977class}. Even worse, these methods hardly clarify the (essential) uniqueness of optimal objective values.


Bearing the motivations and challenges mentioned above, we summarize our contributions as follows:

$\mbox{\ensuremath{\rhd}}$ We propose, in Section~\ref{sec_problemformulation}, two iPOP problems that are optimal~$(0,1)$-matrix completion problems with majorization ordered objectives, arising in EV charging, portfolio optimization, and secure data storage.

$\mbox{\ensuremath{\rhd}}$ We characterize the essential uniqueness of all optimal objective values and identify two particular ones of interesting features for each proposed iPOP problem in Section~\ref{sec_uniqueness}. Meanwhile, we relate our iPOP to optimization over lattices~\cite{zimmermann2011linear}. 

$\mbox{\ensuremath{\rhd}}$ We respectively develop a ``peak-shaving'' and a ``valley-filling'' algorithm in Section~\ref{sec_solutions} to construct an optimal~$(0,1)$-matrix associated with an arbitrary optimal objective value for each of the two iPOP problem. The algorithms are efficient and insightful, e.g., having linear time complexities and uncovering the inherent symmetry of the two problems.


$\mbox{\ensuremath{\rhd}}$ We verify the efficiency of our approach compared to the order-preserving scalarization in Section~\ref{sec_comparison}. Moreover, we present natural extensions of the two iPOP problems, further showing the strength of our solution method.

In Section~\ref{sec_relatedwork}, we introduce the related work and preliminaries. In Section~\ref{sec_conclusion}, we conclude this paper and show future directions. For fluency, we defer most proofs to the Appendix. 

Our study of iPOP can be traced back to the conference paper~\cite{mo2017coordinating}, touching the uniqueness of optimal objective values and an optimal solution algorithm for one iPOP problem (see Problem~\eqref{basicmajmatrix2} herein) without proofs. This paper completes the study of the iPOP problem and complements it with an additional one (see Problem~\eqref{basicmajmatrix1}). We derive the main results for the additional problem and unravel the inherent symmetry of the two iPOP problems. We also augment the study of attainable and feasible sets by characterizing particular (optimal) objective values and solutions. We further prove that our approach can find all optimal objective values instead of one and add extensive comparisons to conventional methods.


\subsubsection*{Notation}
Let~$\underline{n}$ denote the index set~$\{1,\ldots,n\}$, for~$n\in \mathbb{N}$. Let~$\bm 0$~or~$\bm 1$ respectively denote a vector of all zero or one elements. Let~$\bm {e_i}$ be a $(0,1)$-vector of all zero elements except the~$i$th one. Each of~$\bm 0$, $\bm 1$, and~$\bm {e_i}$ has a compatible dimension. For an index sequence~$\left(1,\ldots,m\right)$, a permutation~$\sigma: \underline{m}\rightarrow \underline{m}$ is a bijective function which rearranges the original sequence. Given~$\bm x =[x_{1}~\cdots~x_{n}]'$, we denote its nonincreasing rearrangement by~$\bm x^{\downarrow}=\big[x_{[1]}~\cdots ~x_{[n]}\big]'$, where~$x_{[1]}\geq  \cdots \geq x_{[n]}$. Also, define~$\mathbb{R}^n_{\downarrow}=\left\{\bm x^{\downarrow}\mid\bm x\in \mathbb{R}^n\right\}$ and~$\mathbb{N}^n_{\downarrow}=\left\{\bm x^{\downarrow}\mid\bm x\in \mathbb{N}^n\right\}$. For two vectors of the same length, $\bm x$ and~$\bm y$, let~$\bm x+\bm y$ and~$\bm x - \bm y$ denote the elementwise addition and subtraction, respectively. 
The {H\"{o}lder}~$1$-norm of~$\bm x\in \mathbb{R}^n$ is denoted by~$\|\bm x\|_1$, which sums the absolute values of all the elements in~$\bm x$, namely,~$\|\bm x\|_1=\sum_{i=1}^{n}|x_i|$. The indicator function~$\mathds{1}(\cdot)$ maps an assertion to one if it is true and zero otherwise. For a matrix~$A=[a_{ij}]\in \mathbb{R}^{m\times n}$, its transpose is denoted by~$A'$, while its~$i$th row and~$j$th column are respectively specified by~$\bm a_{i\bigcdot}\in \mathbb{R}^{1\times n}$ and $\bm a_{\bigcdot j}\in \mathbb{R}^{m}$. The vectorization of a matrix~$A$, denoted by $\textbf{vec}(A)$, is the $mn$-dimensional vector obtained by stacking the columns of~$A$:~$\textbf{vec}(A)=\begin{bmatrix}\bm a_{\bigcdot 1}'&\bm a_{\bigcdot 2}'&\cdots&\bm a_{\bigcdot n}'\end{bmatrix}'.$

\section{Related Works and Preliminaries}\label{sec_relatedwork}
\subsection{$(0,1)$-Matrix Completion}
Most results in the literature generalize the seminal Gale-Ryser works~\cite{gale1957theorem,ryser1957combinatorial} by focusing on the matrix feasibility problem with additional or modified constraints. A typical practice is to prescribe certain zeros in addition to given row/column sums. For example, the work~\cite{brualdi2003matrices} considers a fixed zero block. 
The Fulkerson-Chen-Anstee theorem addresses the case where each column has at most one prefixed zero~\cite{fulkerson1960zero,chen1966realization,anstee1982properties}. {Our previous results involve the cases where unassigned positions form a staircase~\cite{chen2016constrained}, a banded~\cite{mo2020market}, or even an arbitrary pattern~\cite{mo2022Tensor}.} By further requiring no two consecutive~$1$'s in every column, the authors of~\cite{nelson20150} derived an existence condition with a series of majorization inequalities, including the single one in the Gale-Ryser theorem. Moreover, the two papers~\cite{fulkerson1959network,mo2018staircase} examine the case with bounded row/column sums, instead of exact ones, while the paper~\cite{brualdi2019gale} studies~$(0,1)$-matrices with given row and column sums modulo~$k$. Unlike these existing results, our work extends the feasibility study by Gale and Ryser to the challenging yet useful optimization study~--~POP.

{We note that several~$(0,1)$-matrix feasibility problems~\cite{chen1966realization,erdosgallai1960degree,berger2011dag} are relevant to graph realization, laying  the foundation for control-related applications like network generation, controllability, and synchronization~\cite{olshevsky2013degree,hsu2019laplacian,siami2016fundamental}. We envision that the advanced optimization study will provide a more innovative perspective on these applications than the feasibility study.}
\subsection{Partial Order Programming}
A binary relation $\preccurlyeq$ on a set $\mathcal{S}$ is a partial order if it is~$\mathit{1}$)~reflexive:~{$x\preccurlyeq x$}, $\mathit{2}$)~transitive: $x\preccurlyeq y $ and~$y\preccurlyeq z $ imply~{$x\preccurlyeq z$}, and~$\mathit{3}$)~antisymmetric:~$x\preccurlyeq y $ and~$y\preccurlyeq x $ imply~$x = y$, for all~$x, y, z\in \mathcal{S}$. Unlike in a total order, two elements in a partially ordered set, or a poset~$(\mathcal{S}, \preccurlyeq)$ may not always be comparable (neither~$x \preccurlyeq y$ nor~\mbox{$y \preccurlyeq x$}). Here is a typical poset~$(\mathbb{R}^n,\leq)$: for~$\bm x,\bm y \in \mathbb{R}^n$, we write~\mbox{$\bm x \leq \bm y$}, saying~$\bm x$ is no more than~$\bm y$ in the elementwise order, if~\mbox{$x_i\leq y_i$}, for all~$i\in \underline{n}$.

POP refers to optimization problems whose objectives are partially ordered; moreover, the optima to a POP problem are those that no other can majorize and may not be unique. An earlier study of POP is vector optimization~\cite[Section 4.7]{boyd2004convex}, where the objective values that are optimal under a cone-induced partial order are usually non-unique and form a Pareto frontier. Another example is the classic stable matching problem originally studied by Gale and Shapley\cite{gale1962college}. While the seminal Gale-Shapley algorithm can find one stable matching, Knuth and other researchers further described all stable matchings by a partial order and pinpointed the optimality of the classic algorithm under the partial order, namely, it generates the (unique) optimal one (the best for all women or all men)~\cite{knuth1997stable}. Analogously, we uplift the seminal works of Gale and Ryser to POP problems with the majorization order. 


\subsection{Majorization}
\begin{definition}\label{defmaj}
Given~$\bm x, \bm y \in \mathbb{R}^n$, we write
\begin{enumerate}
  \item[$\mathit{1}$)] $\bm x \prec_w \bm y$, saying that~$\bm x$ is weakly submajorized by~$\bm y$, if~$\sum_{i=1}^{k}\!x_{[i]}\leq  \sum_{i=1}^{k}y_{[i]}$, for all~$k\in \underline{n}$;
  \item[$\mathit{2}$)] $\bm x \prec^w \bm y$, saying that~$\bm x$ is weakly supermajorized by~$\bm y$, if~$\sum_{i=k}^{n}x_{[i]}\geq  \sum_{i=k}^{n}y_{[i]}$, for all~$k\in \underline{n}$;
  \item[$\mathit{3}$)] $\bm x \prec \bm y$, saying that~$\bm x$ is majorized by~$\bm y$, if~$\bm x \prec_w \bm y$ and~$\bm x \prec^w \bm y$ together.
\end{enumerate}
\end{definition}


Majorization is a powerful tool in many applications~\cite{olkin2016inequalities}. Apart from the existence of a~$(0,1)$-matrix with given line sums~\cite{gale1957theorem,ryser1957combinatorial}, researchers use majorization to characterize the existence of a {series} of pairwise disjoint partial transversals
~\cite{mirsky1966systems}, the conditions for a set of polynomials to be the invariant polynomials of a linear time-invariant system with state feedback~\cite[Section~4.2]{flamm1980new,rosenbrock1970state}, the convergence analysis of distributed Kalman filtering~\cite{del2011majorization}, the networked stabilizability~\cite{liu2018stabilization,chen2018majorization}, and optimal strategies for remote estimation~\cite{nayyar2013optimal,chakravorty2019remote}. Our work, which uses majorization to evaluate the objective values of optimal~$(0,1)$-matrix completion, facilitates the generalization of majorization inequalities to POP with majorization ordered objectives. Particularly, {there are two studies~\cite{kuvcera2022assignment} and~\cite{kim1998simple} dealing with optimization via majorization}. The former concerns a combinatorial problem whose solution set is a {poset of nonnegative integer vectors} with a given sum, ordered by majorization. Meanwhile, the latter examines a class of optimization problems whose optimization criterion is given by two majorization inequalities. Later, we shall show that we are essentially studying optimization over a majorization ordered {poset} with majorization ordered objectives and generalize the study of~\cite{kim1998simple}.

Strictly speaking, majorization is just a preorder in~$\mathbb{R}^n$, not respecting antisymmetry. However, we can bridge the gap by defining the following equivalence relation and canonical set.
\begin{definition} \label{defeql}
For $\bm x, \bm y \in \mathbb{R}^n$, we say $\bm x$ is equivalent to $\bm y$, writing $\bm x \sim \bm y$, if~$\bm x \prec \bm y$ and~$\bm y\prec \bm x$. For a subset~$\mathcal{X}$  of~$\mathbb{R}^n$, we define its canonical set as~$\mathcal{X}_\downarrow=\{\bm x^\downarrow\mid \bm x\in \mathcal{X}\}$.
\end{definition}

We see that~$\bm x \sim \bm y$ if and only if~$\bm x^\downarrow=\bm y^\downarrow$, or in other words, they are rearrangements of each other. 
We see that majorization is a partial order in a canonical set. 

\section{Problem Formulation and Applications}\label{sec_problemformulation}
In this section, we shall propose two optimal~$(0,1)$-matrix completion problems with majorization ordered objectives from EV charging~\cite{chen2016constrained,mo2020market,negrete2016rate}. They generalize the seminal study by Gale and Ryser from feasibility to optimality in iPOP. We also present two illustrative applications.

{
\subsection{From EV Charging to Optimal~$(0,1)$-Matrix Completion}
}
Consider a finite time horizon evenly segmented into~$n$ time slots, while the charging rates of EVs are uniform at one unit per time slot. There come~$m$ EVs, the~$i$th of which requires to be charged in~$r_i$ out of~$n$ slots, for~$i\in \underline{m}$, and we call~$\bm r \in \mathbb{N}^m$ the~\emph{duration profile} as in~\cite{negrete2016rate}. Then, a coordination of the EVs can be denoted by an~\mbox{$m\times n$~$(0,1)$}-matrix~$A$, where~\mbox{$a_{ij}=1$} means that the~$i$th EV gets charged at time slot~$j$ and we require~\mbox{$\|\bm a_{i\bigcdot}\|_1=r_i$}, for all~$i\in\underline{m}$, to fully charge the EVs. The coordination matrix is usually not unique and we aim to find the best one in a reasonable sense. Next, we present two optimality criteria separately for the optimal EV coordination, {leading to two optimal~$(0,1)$-matrix completion problems. } 

{
In the first case, we have the~\emph{supply profile}~$\bm c\in \mathbb{N}^n$, where~$c_j$ denotes the number of available electrical energy units at time slot~$j$, for all~$j\in \underline{n}$. We aim at a flat~\emph{remaining supply profile}~(i.e.,~the difference between the supply profile~$\bm c$ and the column sum vector of a coordination matrix). 
Therefore, given~$\bm c\in \mathbb{N}^n$ and~$\bm r \in \mathbb{N}^m$, we formulate the first optimal~$(0,1)$-matrix completion problem as

\noindent \begin{equation}
\begin{split}\label{basicmajmatrix1}
\underset{A}{\text{~~minimize}_{\prec}} &~~~~\bm c-\sum_{i=1}^{m}\bm a_{i\bigcdot}'\ \\
\text{subject to~}&~~~~A\in \{0,1\}^{m\times n};\;\|\bm a_{i\bigcdot}\|_1=r_i, \forall i\in\underline{m};\\
&~~~~ \|\bm a_{\bigcdot j}\|_1\leq c_j, \forall j\in \underline{n}. 
\end{split}
\end{equation}

By writing~$\text{minimize}_{\prec}$, we aim to find a minimal element or the set of minimal elements among attainable objective values.

In the following, we clarify why we expect a flat remaining supply profile and obtain a more compact form of the above iPOP problem. To this end, we review several useful concepts and the Gale-Ryser theorem~\cite[Section~7]{olkin2016inequalities}. 

\begin{definition}
  The partition conjugate of a vector~$\bm x\in \mathbb{N}^n$ is a vector~$\bm x^*$, whose $j$th element is the number of elements no less than~$j$ in~$\bm x$, namely,~$x_j^*=\sum\nolimits_{i=1}^n\mathds{1}(x_i\geq j)$.
\end{definition}
\begin{figure}[t]
\centering
\includegraphics[scale=0.36]{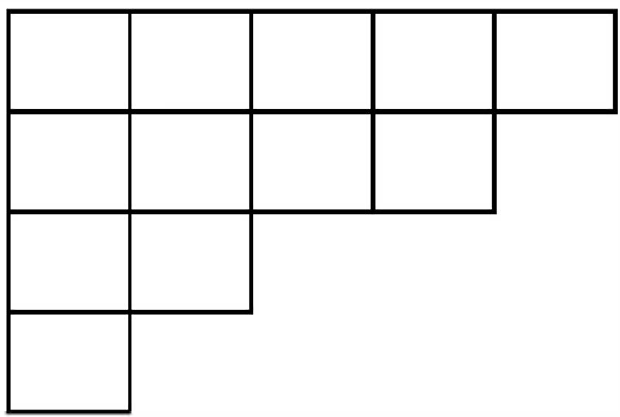}
\caption{The Young diagram regarding~$x=[5~4~2~1]'$ and its partition conjugate~$x^*=[4~3~2~2~1~0~\cdots~0]'$.}\label{fig: Youngdi}
\end{figure}

We see~$x^*_j=0$ if $j>x_{[1]}$. By adjusting the number of zeros, we give~$\bm x^*$ a required dimension that is no less than~$x_{[1]}$. In Fig.~\ref{fig: Youngdi}, we graphically illustrate the partition conjugate, which relates to a Young diagram~\cite{young1901quantitative} consisting of a collection of left-justified cells of equal size. Moreover, the number of cells in each row corresponds to each element of~$\bm x$, while that of each column corresponds to each element of~$\bm x^*$.


Let~$\mathcal{A}(\bm r,\bm {x})$ denote the set of~$m\times n$~$(0,1)$-matrices whose row and column sum vectors are~\mbox{$\bm r\in\mathbb{N}^m$} and~$\bm {x}\in \mathbb{N}^n$, respectively. The Gale-Ryser theorem states that~$\mathcal{A}(\bm r,\bm {x})$ is nonempty if and only if $\bm {x}\prec \bm r^*$~\cite{gale1957theorem,ryser1957combinatorial}, suggesting that~\mbox{$\left\{\bm r \mid \mathcal{A}(\bm r,\bm{y})\neq \emptyset \right\}\subseteq \left\{\bm r \mid \mathcal{A}(\bm r,\bm{x})\neq \emptyset \right\}$} for~$\bm{x},\bm{y}\in \mathbb{N}^n$ and~$\bm{x}\prec \bm{y}$. Thus, requiring the remaining supply profile to be flat enables us to accommodate more new EVs~\cite{mo2020market,negrete2016rate}. 
More interestingly, this requirement also suggests charging the EVs at the slots with more supplies~(lower prices), consistent with the ``peak-shaving'' behavior in smart grids~\cite{chakrabortty2011control}, which will be made more clear later.

In addition, the Gale-Ryser theorem implies that~\mbox{$\bm {x}\in \mathbb{N}^n$} is the column sum vector of a coordination matrix if and only if~\mbox{$\bm {x}\prec \bm r^*$}. Hence, the optimal remaining supply profiles are the minimal elements of~\mbox{$\{\bm c-\bm {x}\mid \bm {x} \in \mathbb{N}^n, \bm {x} \prec \bm r^*\text{, and }\bm {x} \leq \bm c\}$} under majorization. Then, we can reformulate Problem~(\ref{basicmajmatrix1}) as
\begin{equation}
\begin{split}\label{basicmaj1}
\underset{\bm {x}}{\text{~~minimize}_{\prec}} &~~~~\bm c-\bm {x}\ \\
\text{subject to~}&  ~~~~\bm {x}\in \mathbb{N}^n, \bm {x} \prec \bm r^*\text{, and } \bm {x} \leq \bm c.
\end{split}
\end{equation}

}

In the second case, we consider base loads like road lamps, consuming~$b_j\in \mathbb{N}$ units of electrical energy at slot~$j$, for all~$j\in \underline{n}$. We aim to coordinate the EV charging so that the combined power consumption of the base loads and EVs is as smooth as possible for the following reasons~\cite{luh1982load,chakrabortty2011control,khemakhem2019impact,casini2021chance,mo2021infocom}. First, peaks in demand increase the infrastructure costs by requiring additional generators and ramping capacities. Second, a fluctuant load may aggravate emission costs, voltage deviations, and power losses. Finally, the changes in demands over time usually intensify undesirable market volatility. Also, we show that majorization can properly evaluate the smoothness level of the power consumption from two aspects. First, the smaller a vector is under majorization, the smaller the variance of the vector is~\cite[Section~1.C]{olkin2016inequalities}. Second, using majorization as a measure of smoothness generalizes the~\emph{valley-filling} behavior, whose core is shifting coordinated loads like EV charging to time slots with lower existing load and higher network capacity~\cite{chakrabortty2011control}. Thus, given~$\bm b\in \mathbb{N}^n$ and~$\bm r \in \mathbb{N}^m$, we formulate the second optimal~$(0,1)$-matrix completion problem as
\begin{equation}
\begin{split}\label{basicmajmatrix2}
\underset{A}{\text{~~minimize}_{\prec}} &~~~~\bm b+\sum_{i=1}^{m}\bm a_{i\bigcdot}'\ \\
\text{subject to~}&~~~~A\in \{0,1\}^{m\times n};\;\|\bm a_{i\bigcdot}\|_1=r_i, \forall i\in\underline{m}.
\end{split}
\end{equation}

Note that in the above formulation, we do not enforce inequality constraints on the column sums as in Problem~\eqref{basicmajmatrix1}. By the Gale-Ryser theorem, we note that the smoothest combined consumption profiles correspond to the minimal elements of~\mbox{$\{\bm b+\bm {x}\mid \bm {x} \in \mathbb{N}^n\text{ and }\bm {x}\prec \bm r^*\}$} under majorization. Similarly, we can reformulate the above iPOP problem into a more compact form:
\begin{equation}
\begin{split}\label{basicmaj2}
\underset{\bm {x}}{\text{~~minimize}_{\prec}} &~~~~\bm b+\bm {x}\ \\
~~~~~~\text{subject to~}&  ~~~~\bm {x}\in \mathbb{N}^n \text{ and } \bm {x} \prec \bm r^*.~~~~~~~~~~~~~~~
\end{split}
\end{equation}

So far, we have formulated the two iPOP problems of our interest. We assume~$\bm b,\bm c\in \mathbb{N}^n_{\downarrow}$ unless specified otherwise, respectively called a base and a ceiling vector. {We do not assume the order of elements in the row sum vector~$\bm r$.} The two iPOP problems seem to be symmetric in terms of elementwise addition and subtraction (see the objectives of Problems~(\ref{basicmaj1}) and~(\ref{basicmaj2})), so we shall study them synchronously and wonder whether there are dual properties between them in later sections. Particularly, we prove the essential uniqueness of optimal objective values, which facilitates finding all the optimal solutions. Moreover, we propose a specialized approach to each proposed optimal~$(0,1)$-matrix completion problem with majorization ordered objectives, agreeing with the rule of thumb ``peak-shaving'' or ``valley-filling.'' The resulting algorithms provide a fine-grained perspective on the inherent symmetry of the two iPOP problems.
\subsection{More Illustrative Applications}
There are many applications of the proposed problems in the literature, e.g., load assignment in crossbar switches~\cite{kim1998simple} and biological sequence analysis~\cite{wan2000global}. Next, we present two more, which directly lead to the compact reformulations (namely, Problems~(\ref{basicmaj1}) and~(\ref{basicmaj2})). This fact corroborates the significance of our study beyond optimal matrix completion. 

\emph{$\mathit{1}$) Portfolio optimization.} A broker sells $n$ kinds of assets, which are functional substitutes of each other. The ceiling vector $\bm c$ discloses the original asset inventory of the broker. An investor asks for a portfolio of assets given his/her risk tolerance. Denote the portfolio by $\bm {x}$, where~${x}_k$ is the quantity of the $k$th asset in this portfolio. The broker partitions the available assets into two parts,~$\bm {x}$ and~$\bm c-\bm {x}$, so that the risk tolerance of the investor is satisfied and the broker undertakes as little risk from the remaining assets as possible~\cite{markowitz2010portfolio}. The proverb says that you should not put all your eggs in one basket. Following this idea, we quantitatively describe the risk tolerance as, at most, how many bad assets the broker/investor is willing to possess in the case that there are $k$ kinds of assets (bad assets) suffering deep losses, for all~$k\in \underline{n}$~\cite{egozcue2010gains,ogryczak2002dual}. In the worst case, the kinds of bad assets are exactly those reserved/bought the most by the broker/investor. Thus, the investor's risk tolerance can be described by a threshold vector~$\bm r^*$, while the risk of the remaining assets corresponds to~$\bm c-\bm {x}$. Clearly, the broker requires to solve Problem~(\ref{basicmaj1}). 

\emph{$\mathit{2}$) Secure data storage.} A direct way to keep a confidential file safe is to break it into distinct data centers~\cite{redlich2006data,robling1982cryptography,yue1973optimality}. Assume that there are~$n$ data centers, and the current loads of these data centers are described by a base vector~$\bm b$. In the first step, we properly encrypt a confidential file in such a way that the total information is divided into several pieces of data and each piece contains almost the same amount of information. Then, we separate these data pieces into the~$n$ data centers, saying there are ${x}_k$ pieces stored in the $k$th center. The security requirement of the data scattering for the file can be described by a threshold vector~$\bm r^*$, which reveals at most how many pieces of data from the file can be exposed, for all~$k\in \underline{n}$, provided that $k$ data centers are attacked. Apart from fulfilling the security requirement, the storage service provider should balance the loads of the~$n$ data centers. In other words, it expects that the storage profile~$\bm {x}$ is majorized by the security threshold vector~$\bm t$ and the combined load profile~$\bm b+\bm {x}$ is as smooth as possible. Thus, such a secure data storage problem can be mathematically formulated as Problem~(\ref{basicmaj2}).

\section{Uniqueness of Optimal Objective Values}\label{sec_uniqueness}
The optimal objective value of a POP problem is usually not unique, and it is challenging to identify all the optimal objective values for an iPOP. We show, in this section, that the optimal objective values for each proposed iPOP problem are essentially unique in the sense that they are rearrangements of each other. Note that not all rearrangements of an optimal objective value are attainable; nevertheless, we identify two particular ones of interesting properties for each iPOP problem by analyzing the attainable and feasible sets defined later. These results lay the foundations for the next section, where we develop an efficient and insightful approach, either generating an arbitrary optimal objective value or justifying the infeasibility of the iPOP problem.

\subsection{Attainable Sets and Essential Uniqueness}

Throughout this subsection, we assume Problems~(\ref{basicmajmatrix1}) and~(\ref{basicmajmatrix2}) are feasible. We respectively denote the sets of the column sum vectors of feasible~$(0,1)$-matrices in Problems~(\ref{basicmajmatrix1}) and~(\ref{basicmajmatrix2}) by
\begin{equation*}
\begin{split}
    \mathcal{{X}}^\ominus=\{\bm {x}\in \mathbb{N}^n \mid A\in \{0,1\}^{m\times n};\|\bm a_{i\bigcdot}\|_1=r_i, \forall i\in\underline{m};\\ x_j=\|\bm a_{\bigcdot j}\|_1\leq c_j, \forall j\in \underline{n}\} \text{ and }\\
    \mathcal{{X}}^\oplus=\{\bm {x} \in \mathbb{N}^n\mid A\in \{0,1\}^{m\times n};\|\bm a_{i\bigcdot}\|_1=r_i, \forall i\in\underline{m};\\ x_j=\|\bm a_{\bigcdot j}\|_1, \forall j\in \underline{n}\},
    \end{split}
\end{equation*}
amounting to the feasible sets of Problems~(\ref{basicmaj1}) and~(\ref{basicmaj2}).

Moreover, we respectively denote the attainable sets of Problems~(\ref{basicmajmatrix1}) and~(\ref{basicmajmatrix2}) by $$\mathcal{V}^\ominus=\left\{\bm c - \bm {x} \mid \bm {x} \in \mathcal{{X}}^\ominus\right\} \text{ and }\mathcal{V}^\oplus=\left\{\bm b + \bm {x} \mid \bm {x} \in \mathcal{{X}}^\oplus \right\}.$$ Accordingly, their canonical attainable sets are respectively
$$\mathcal{V}^\ominus_\downarrow\!=\!\left\{(\bm c - \bm {x})^\downarrow \!\mid \bm {x} \in \mathcal{{X}}^\ominus\right\} \text{ and }\mathcal{V}^\oplus_\downarrow\!=\!\left\{(\bm b + \bm {x})^\downarrow \!\mid \bm {x} \in \mathcal{{X}}^\oplus \right\}.$$

Before proceeding, we first examine the relationship between the attainable and canonical attainable sets. Then, we derive Proposition~\ref{propcan}, which implies the existence of a particular optimal objective value regarding the order of elements.
\begin{proposition}\label{propcan}
Suppose~$\bm b,\bm c\in \mathbb{N}^n_{\downarrow}$.
\begin{enumerate}
    \item[$\mathit{1}$)] If~$\bm v \in \mathcal{V}^\ominus$, then~{$\bm v^\downarrow \in \mathcal{V}^\ominus$}.
    \item[$\mathit{2}$)] If~$\bm v \in \mathcal{V}^\oplus$, then $\bm v^\downarrow \in \mathcal{V}^\oplus$.
\end{enumerate}
\end{proposition}
Proposition~\ref{propcan}, proven in Appendix~\ref{PMPOMapdA4}, states that~$\mathcal{V}^\ominus_\downarrow\subseteq \mathcal{V}^\ominus$ and~$\mathcal{V}^\oplus_\downarrow\subseteq \mathcal{V}^\oplus$, when~$\bm b,\bm c\in \mathbb{N}^n_{\downarrow}$. Such inclusion relations suggest that there exists a particular optimal objective value whose elements have the same order as those of~$\bm b$ in Problem~(\ref{basicmajmatrix2}) or~$\bm c$ in Problem~(\ref{basicmajmatrix1}). Since~$\bm b,\bm c\in \mathbb{N}^n_{\downarrow}$, the particular optimal objective value is in its nonincreasing rearrangement. 

Before characterizing the essential uniqueness of optimal objective values, let us clarify several concepts. Note that the attainable sets are preordered by majorization, while the canonical attainable sets are partially ordered by majorization. In a preordered set or poset, we differentiate a minimal element from a least one due to the possible non-comparability. Formally, an element is said to be minimal if it does not majorize another in the preordered set or poset; moreover, it is said to be a least element if all others in the set majorize it. The maximal or greatest element is defined similarly. The least element in a preordered set or a poset may not exist. If existing, the least element in a poset is unique. The least ones in a preordered set may not be unique; however, they are said to be essentially unique by defining a simple equivalence relation as in Definition~\ref{defeql}. Precisely, the least elements in a set preordered by majorization, if existing, are essentially unique in the sense that they are rearrangements of each other.

Recall that the minimal elements of~$\mathcal{{V}}^\ominus$ (or~$\mathcal{{V}}^\oplus$) constitute the optimal objective values of Problem~(\ref{basicmajmatrix1}) (or Problem~(\ref{basicmajmatrix2})). We wonder whether the optimal objective values are the least elements of the corresponding attainable set. The following theorem answers this critical question affirmatively.
\begin{theorem}
  If~$\bm{u}$ and~$\bm{v}$ are minimal in~$(\mathcal{V}^\ominus,\prec)$, then they are the least elements of~$(\mathcal{V}^\ominus,\prec)$~and~$\bm{u}\sim \bm{v}$. The same is true if we replace~$\mathcal{V}^\ominus$ with~$\mathcal{V}^\oplus$.
  \label{leastthm}
\end{theorem}
Theorem~\ref{leastthm}, proven in Appendix~\ref{PMPOMapdA6}, pinpoints that all the optimal objective values of each iPOP problem are essentially unique because they share the same non-increasing rearrangement. In general, not all rearrangements of an optimal objective value are attainable. Combining Proposition~\ref{propcan} and Theorem~\ref{leastthm}, we obtain the following theorem.

\begin{theorem}\label{leastthmcan}
 The least element of~$(\mathcal{V}^\ominus_\downarrow,\prec)$ or~$(\mathcal{V}^\oplus_\downarrow,\prec)$ exists, respectively. Moreover, it is an optimal objective value of Problem~(\ref{basicmajmatrix1}) or Problem~(\ref{basicmajmatrix2}) when~$\bm b,\bm c\in \mathbb{N}^n_{\downarrow}$.
\end{theorem}



Theorem~\ref{leastthmcan} explicitly clarifies that the particular optimal objective value mentioned before is the least element of the canonical set for Problem~(\ref{basicmajmatrix1}) or Problem~(\ref{basicmajmatrix2}) when~\mbox{$\bm b,\bm c\in \mathbb{N}^n_{\downarrow}$}.
Note that the least element in~$(\mathcal{V}^\ominus_\downarrow,\prec)$ or~$(\mathcal{V}^\oplus_\downarrow,\prec)$ is unique and solely determined by the setup~$(\bm c,\bm r)$ or~$(\bm b, \bm r)$, respectively. Motivated by this fact, we respectively define a ``subtraction'' operation and an ``addition'' operation.
\begin{definition}
  The least elements of~$(\mathcal{V}^\ominus_\downarrow,\prec)$ and~$(\mathcal{V}^\oplus_\downarrow,\prec)$ are defined as~$\bm c \ominus \bm r$ and~$\bm b \oplus \bm r$, respectively.
\end{definition}

We observe that~$\bm c \ominus \bm r$ or~$\bm b \oplus \bm r$ is well-defined if Problem~(\ref{basicmajmatrix1}) or Problem~(\ref{basicmajmatrix2}) is feasible.  The two operations will facilitate the interpretations of our solution approach in the next section. 

\subsection{Feasible Sets and Lattices}
This subsection is for the feasible sets of our iPOP problems and analyzing their structural properties. These results give us more insights into our iPOP problems and benefit future studies. Notably, we prove that our iPOP problems belong to the optimization over a lattice that has attracted attention in system design~\cite{kuvcera2022assignment} and code construction~\cite{parker1999construction}. Meanwhile, we identify another particular optimal objective value, the elements in whose corresponding column sum vector respectively have the same order as those in~$\bm c$ for Problem~(\ref{basicmajmatrix1}) or the reverse order as those in~$\bm b$ for Problem~(\ref{basicmajmatrix2}). This fact shows that our iPOP problems usually have multiple optimal objective values, and it is challenging to identify them all. Nevertheless, we develop a solution approach in the next section to find all optimal objective values, including the two particular ones (See Proposition~\ref{propcan} before and Proposition~\ref{rescanon} later). As a by-product, the proposed algorithms also provide a numerical way to check the feasibility of the considered iPOP problems.

We see that Problem~(\ref{basicmajmatrix2}) or Problem~(\ref{basicmaj2}) is always feasible. Next, let us study the feasibility of Problem~(\ref{basicmajmatrix1}) or Problem~(\ref{basicmaj1}).
\begin{proposition}
  Problem~(\ref{basicmajmatrix1}) or Problem~(\ref{basicmaj1}) is feasible if and only if~$\bm c\prec^w \bm r^*$, or equivalently,~$\bm r \prec_w \bm c^*$.\label{feacond}
\end{proposition}
By Problem~(\ref{basicmaj1}), Proposition~\ref{feacond} shows how the elementwise order influences a majorization ordered set involving nonnegative integers. It is essentially an analogy of an existing result involving real numbers (see Lemma~\ref{realfeacondcombineprecminusmaj} in Appendix~\ref{PMPOMapdA0}), which was originally reported in~\cite[Section~5.A]{olkin2016inequalities}. The proposition follows from the Gale-Ryser theorem and the adequacy theorem in~\cite{negrete2016rate} focusing on Problem~\eqref{basicmajmatrix1}. We give a more concise proof directly concerning Problem~(\ref{basicmaj1}) in Appendix~\ref{PMPOMapdA2}.

Similar to before, we consider the canonical feasible sets of Problems~(\ref{basicmaj1}) and~(\ref{basicmaj2}), which are respectively
$$\mathcal{{X}}^\ominus_\downarrow=\left\{\bm {x}^\downarrow \mid \bm {x} \in \mathcal{{X}}^\ominus\right\}~\text{and}~\mathcal{{X}}^\oplus_\downarrow=\left\{\bm {x}^\downarrow \mid \bm {x} \in \mathcal{{X}}^\oplus \right\}.$$

Furthermore, to characterize a particular feasible solution, we  define~$\mathcal{{X}}^\oplus_\uparrow=\left\{\bm {x}^\uparrow \!\mid \bm {x} \in \mathcal{{X}}^\oplus\!\right\}$, where~$\bm {x}^{\uparrow}$ denotes the nondecreasing rearrangement of~$\bm {x}$. Note that there is a homomorphism between~$\mathcal{{X}}^\oplus_\downarrow$ and~$\mathcal{{X}}^\oplus_\uparrow$. The proposition below, proven in Appendix~\ref{PMPOMapdA3}, shows that, to obtain an optimal solution, it suffices to consider the feasible solutions in the nonincreasing rearrangement for Problem~(\ref{basicmaj1}) or nondecreasing rearrangement for Problem~(\ref{basicmaj2}).
\begin{proposition}
Suppose~$\bm b,\bm c\in \mathbb{N}^n_{\downarrow}$.
\begin{enumerate}
  \item[$\mathit{1}$)] If $\bm {x} \in \mathcal{{X}}^\ominus$, then $\bm {x}^\downarrow\in \mathcal{{X}}^\ominus$ and $\bm c-\bm {x}^\downarrow \prec \bm c-\bm {x}$.
  \item[$\mathit{2}$)] If~$\bm {x} \in \mathcal{{X}}^\oplus$, then $\bm {x}^\uparrow\in \mathcal{{X}}^\oplus$ and $\bm b+\bm {x}^\uparrow \prec \bm b+\bm {x}$.
\end{enumerate}
\label{rescanon}
\end{proposition}
By Proposition~\ref{rescanon}, we have~$\mathcal{X}^\ominus_\downarrow\subseteq \mathcal{X}^\ominus$ and~$\mathcal{X}^\oplus_\uparrow\subseteq \mathcal{X}^\oplus$, when~$\bm b,\bm c\in \mathbb{N}^n_{\downarrow}$. Accordingly, we conclude that there exists a particular optimal objective value, the elements in whose corresponding column sum vector respectively have the same order as those in~$\bm c$ for Problem~(\ref{basicmajmatrix1}) or the reverse order as those in~$\bm b$ for Problem~(\ref{basicmajmatrix2}).
Such particular optimal objective value may not be unique, as exemplified in the next section.

Recall that~$\mathcal{{X}}^\ominus_\downarrow$ and~$\mathcal{{X}}^\oplus_\uparrow$ are partially ordered by majorization. We next show that they have more subtle structures. To this end, let us clarify several necessary preliminaries.

For a poset~$(\mathcal{S}, \preccurlyeq)$ and a subset~$\mathcal{T}$ of~$\mathcal{S}$, we respectively use~$\inf_{\mathcal{S}}\mathcal{T}$ and~$\sup_{\mathcal{S}}\mathcal{T}$ to denote the infimum and supremum of~$\mathcal{T}$ in~$(\mathcal{S}, \preccurlyeq)$. If~$\mathcal{T}$ consists of only two elements, the infimum and supremum are, respectively, called the meet and join of the two elements. A poset is called a lattice if every pair of elements has a meet~and~a join~\cite{birkhoff1940lattice}. Given a subset~$\mathcal{T}$ of a set~$\mathcal{S}$, we say that~$\mathcal{T}$ is a sublattice of~$\mathcal{S}$ under a partial order~$\preccurlyeq$ if~$(\mathcal{S}, \preccurlyeq)$ and~$(\mathcal{T}, \preccurlyeq)$ are both lattices, and~for every pair of elements~$x,y$ in~$\mathcal{T}$, it holds that~$\sup_{\mathcal{T}}\{x,y\}=\sup_{\mathcal{S}}\{x,y\}$ and~$\inf_{\mathcal{T}}\{x,y\}=\inf_{\mathcal{S}}\{x,y\}$~\cite{szsz1964introduction}. According to the definition, a poset that is both a lattice and a subset of a larger lattice may not necessarily be a sublattice of the larger lattice, different from the concepts of linear subspaces and subgroups.

It is clear that majorization is a partial order on~$\mathbb{R}^n_{\downarrow}$. Neither~$\mathbb{R}^n_{\downarrow}$ nor~$\mathbb{N}^n_{\downarrow}$ is a lattice under majorization, but for~$\tau\in \mathbb{R}$, the set~\mbox{$\mathcal{R}_{\tau}=\big\{\bm x\in \mathbb{R}^{n}_{\downarrow}\mid \sum_{i=1}^{n}x_i=\tau\big\}$} is proven to be a lattice under majorization~\cite{cicalese2002supermodularity}. A partition of a nonnegative integer~$\tau$ is a sequence of nonnegative integers whose sum is~$\tau$. For notational convenience, we define the partition set~$\mathcal{N}_\tau=\big\{\bm x\in \mathbb{N}^{n}_{\downarrow}\mid \|\bm x\|_1=\tau \big\}$, and the majorization order regarding~$\mathcal{N}_\tau$ is also known as the dominance order~\cite[Section~3]{james2006representation}.
 Although~$(\mathcal{N}_\tau,\prec)$ is a subset of~$\mathcal{R}_{\tau}$ and a lattice~\cite[Section~5.E]{olkin2016inequalities}, it is not generally a sublattice of~\mbox{$(\mathcal{R}_{\tau},\prec)$} since the join of two distinct elements in~$\mathcal{N}_\tau$ may be different from that in~$\mathcal{R}_{\tau}$. For example, if $\bm x=[5~2~2~2]'$ and $\bm y=[4~3~3~1]'$, then we have $\sup\nolimits_{\mathcal{N}_{11}}\{\bm x, \bm y\}=[5~3~2~1]'$ while~\mbox{$\sup\nolimits_{\mathcal{R}_{11}}\{\bm x, \bm y\}=[5~2.5~2.5~1]'$}. The meet and join of two distinct elements in~$\mathcal{R}_{\tau}$ can be efficiently calculated by the methods in~\cite{cicalese2002supermodularity}, while the methods for calculating those in~$\mathcal{N}_{\tau}$ will be given later. Note that~$\mathcal{{X}}^\ominus_\downarrow$ and~$\mathcal{{X}}^\oplus_\downarrow$ are subsets of~$\mathcal{N}_{\|\bm t\|_1}$, and we wonder whether they inherit the lattice structure from their common superset~$\mathcal{N}_{\|\bm t\|_1}$. See the answer in Proposition~\ref{feaslattice}.

\begin{proposition}
  The canonical feasible set~$\mathcal{{X}}^\ominus_\downarrow$ or~$\mathcal{{X}}^\oplus_\downarrow$ with over two elements is a sublattice of~$\mathcal{N}_{\|\bm t\|_1}$ under majorization.\label{feaslattice}
\end{proposition}

Proposition~\ref{feaslattice}, proven in Appendix~\ref{PMPOMapdAf}, generalizes the result in~\cite{wan2000global} for a special case where~$n=4$ and the elements of~$\bm c$ are large enough. Its proof shows that obtaining~$\inf_{\mathcal{N}_{\tau}}\{\bm x, \bm y\}$ is easier than~$\sup_{\mathcal{N}_{\tau}}\{\bm x, \bm y\}$, since we derive the former from the minima of the corresponding leading partial sums of~$\bm x^\downarrow$ and~$\bm y^\downarrow$, while we cannot necessarily obtain the latter by the maxima accordingly. A similar phenomenon was observed for~\mbox{$\inf_{\mathcal{R}_{\tau}}\{\bm x, \bm y\}$} and~$\sup_{\mathcal{R}_{\tau}}\{\bm x, \bm y\}$~\cite{cicalese2002supermodularity}. From an optimization viewpoint, the reason is that the pointwise maximum of concave functions is not necessarily concave~\cite[Section 3.2]{boyd2004convex}. {Nevertheless, unlike~$\sup_{\mathcal{R}_{\tau}}\{\bm x, \bm y\}$, there is an easier way to attain~$\sup_{\mathcal{N}_{\tau}}\{\bm x, \bm y\}$ via partition conjugates involving~$(0,1)$-matrix completion.} Specifically, we show by simple calculations that~$\sup_{\mathcal{N}_{\tau}}\{\bm x, \bm y\}=\left(\inf_{\mathcal{N}_{\tau}}\{\bm x^*, \bm y^*\}\right)^*$, so calculating~$\sup_{\mathcal{N}_{\tau}}\{\bm x, \bm y\}$ is not harder than~$\inf_{\mathcal{N}_{\tau}}\{\bm x, \bm y\}$ except that we need to calculate the partition conjugates three times.

Proposition~\ref{rescanon} and Proposition~\ref{feaslattice} together indicate that the two iPOP problems essentially amount to the optimization over lattices~$\mathcal{{X}}^\ominus_\downarrow$ and~$\mathcal{{X}}^\oplus_\uparrow$ (homomorphic to~$\mathcal{{X}}^\oplus_\downarrow$). Many interesting results have been derived by virtue of optimization over a lattice~\cite{kuvcera2022assignment,parker1989partial,zimmermann2011linear} like constructing Huffman codes~\cite{parker1999construction}. From this perspective, we can expect more potential applications of our iPOP problems that we may neglect herein.

When~$\bm b$ or~$\bm c$ equals~$x\bm 1$, for a certain~$x\in \mathbb{R}$, we conclude by Proposition~\ref{feaslattice} that the canonical attainable set of Problem~(\ref{basicmajmatrix1}) or Problem~(\ref{basicmajmatrix2}) is also a lattice under majorization. However, the sets~$\mathcal{V}^\ominus_\downarrow$ and~$\mathcal{V}^\oplus_\downarrow$ are not lattices under majorization in general. We exemplify this observation in Appendix~\ref{PMPOMapdA5}. To a certain extent, this observation implies the challenge of solving the proposed problem by traditional methods for POP and calls for efficient specialized solution algorithms to be introduced.

\section{``Peak-Shaving'' and ``Valley-Filling'' Solutions}\label{sec_solutions}
In this section, we propose a ``peak-shaving'' and a ``valley-filing'' algorithm (see Algorithms~\ref{MAlgo1} and~\ref{MAlgo2}) respectively for optimal solutions and objective values of Problem~(\ref{basicmajmatrix1}) and Problem~(\ref{basicmajmatrix2}) (see Theorem~\ref{algthm}). Moreover, by carefully tackling ties, our approach can lead us to every optimal objective value (see Theorem~\ref{alloptobj}). As a by-product, the ``peak-shaving'' algorithm can also be used to check the feasibility of Problem~(\ref{basicmajmatrix1}).


We note that Ryser's algorithm is well-known to be efficient in constructing a $(0,1)$-matrix with given row and column sums~\cite[Section~3]{brualdi2006combinatorial}. It iteratively constructs an admissible matrix row by row or column by column, in a greedy manner. The two algorithms we develop inherit such merits of Ryser's algorithm. Specifically, the peak-shaving algorithm sequentially constructs the rows of an optimal~$(0,1)$-matrix by decreasing largest elements of a vector; in contrast, the valley-filling one does so by increasing smallest elements. Thus, the two algorithms also provide a fine-grained perspective on the inherent symmetry of the two proposed iPOP problems.

\begin{theorem}
  Given~$\bm b, \bm c \in \mathbb{N}^n$ and~$\bm r \in \mathbb{N}^m$, Algorithm~\ref{MAlgo1} and Algorithm~\ref{MAlgo2} respectively generate optimal~$(0,1)$-matrices for Problem~(\ref{basicmajmatrix1}) and Problem~(\ref{basicmajmatrix2}) with the time complexity~$\mathcal{O}(mn)$.\label{algthm}
\end{theorem}

We defer the proof of Theorem~\ref{algthm} to Appendix~\ref{PMPOMapdA7}. Following is an immediate corollary. It shows that, like Ryser's algorithms, Algorithm~\ref{MAlgo1} either generates an admissible matrix or suggests the infeasibility of the considered matrix completion. 
\begin{corollary}
  Problem~(\ref{basicmajmatrix1}) is feasible if and only if the optimal objective value~$\bm{\bar{c}}$ generated by Algorithm~\ref{MAlgo1} is elementwise nonnegative, namely~$\bm{\bar{c}}\geq \bm 0$.
\end{corollary}
Algorithm~\ref{MAlgo1} is easy to implement because it only involves operations like sorting and subtraction. It decomposes the related iPOP problem into a series of sorting processes, since the key operation of each iteration is to find a number of largest elements~(Line~$3$ in Algorithm~\ref{MAlgo1}). We name it the peak-shaving algorithm becasue it iteratively subtracts ones from the columns with more remaining column sums~(Line~$4$-$5$ in Algorithm~\ref{MAlgo1})~\cite{chakrabortty2011control}. Correspondingly, we call Algorithm~\ref{MAlgo2} the valley-filling algorithm because it adds ones to the columns with fewer aggregated column sums in each iteration.


In Algorithm~\ref{MAlgo1} or Algorithm~\ref{MAlgo2}, we tackle ties randomly, following a uniform distribution. Thus, we may obtain different objective values when running the algorithm twice, though they are equivalent and both optimal. If we tackle ties by giving priority to the positions with larger or smaller indices respectively, we can obtain the particular optimal objective value in Proposition~\ref{propcan}. Differently, if we give priority according to the current aggregated column sums, we can obtain a particular optimal objective value in Proposition~\ref{rescanon}. More interestingly, we find out that we can obtain all optimal objective values of Problem~(\ref{basicmajmatrix1}) or Problem~(\ref{basicmajmatrix2}) by enumerating all priority choices whenever encountering ties. We formally state this result as Theorem~\ref{alloptobj} and give the proof in Appendix~\ref{PMPOMapdAall}.
\begin{theorem}\label{alloptobj}
  Each least element in~$(\mathcal{V}^\ominus,\prec)$ or~$(\mathcal{V}^\oplus,\prec)$ can respectively be the output of Algorithm~\ref{MAlgo1} or Algorithm~\ref{MAlgo2} with a positive probability.
\end{theorem}

\begin{algorithm}[t]
\caption{A peak-shaving approach to Problem~(\ref{basicmajmatrix1})}
\label{MAlgo1}
    \KwIn {A ceiling vector~$\bm c\in \mathbb{N}^n_\downarrow$ and~$\bm r\in \mathbb{N}^m$ with~$\bm c\prec^w \bm r^*$.}
    \KwOut  {An optimal objective value~$\bm{\bar{c}}=\bm{\bar{c}}^{(i)} \in \mathcal{V}^\ominus$ and a~$(0,1)$-matrix~$A=[a_{ij}]\in \{0,1\}^{m\times n}$.}
    {{\bf Initialization:} $i=1$, $m=1$, $\bm{\bar{c}}^{(0)}=\bm c$\;}
    \While {$i\leq m-1$}{
    {Identify the positions of the $r_i$ largest elements} {in~$\bm{\bar{c}}^{(i-1)}$. Follow the uniform distribution to tackle ties} such that exactly $r_i$ positions are specified\; {Let~$\bm {c}^{(i)}_{temp}$ be an~$n\times 1$ $(0,1)$-vector whose ones} {appear exactly in the prespecified positions\;}
    {$\bm a_{i\bigcdot} =\bm {c}^{(i)}_{temp}$; $\bm{\bar{c}}^{(i)}=\bm{\bar{c}}^{(i-1)}-\bm c_{temp}^{(i)}$; $i=i+1$\;}
    }
\end{algorithm}
    \begin{algorithm}[t]
\caption{A valley-filling approach to Problem~(\ref{basicmajmatrix2})}
\label{MAlgo2}
    \KwIn {A base vector~$\bm b\in \mathbb{N}^n_\downarrow$ and~$\bm r\in \mathbb{N}^m$.}
    \KwOut {A optimal objective value~$\bm{\bar{b}}=\bm{\bar{b}}^{(i)} \in \mathcal{V}^\oplus$  and a~$(0,1)$-matrix~$A=[a_{ij}]\in \{0,1\}^{m\times n}$.}
    {{\bf Initialization:} $i=1$, $m=1$, $\bm{\bar{b}}^{(0)}=\bm b$\;}
    \While {$i\leq m-1$}{
    {Identify the positions of the $r_i$ smallest elements} {in~$\bm{\bar{b}}^{(i-1)}$. Follow the uniform distribution to tackle ties} such that exactly $r_i$ positions are specified\; {Let~$\bm {b}^{(i)}_{temp}$ be an~$n\times 1$ $(0,1)$-vector whose ones} {appear in the prespecified positions exactly\;}
     {$\bm a_{i\bigcdot} =\bm {b}^{(i)}_{temp}$; $\bm{\bar{b}}^{(i)}=\bm{\bar{b}}^{(i-1)}+\bm b_{temp}^{(i)}$; $i=i+1$\;}
    }
\end{algorithm}

The above theorem suggests an efficient way to obtain all optimal objective values. It remains to show how to find all the~$(0,1)$-matrices with given row and column sums. To our delight, for every two matrices in~$\mathcal{A}(\bm r,\bm {x})$, one can be obtained from another by a series of simple operations, as proved in~\cite{ryser1957combinatorial}. Thus, it is not much harder to obtain all the optimal~$(0,1)$-matrices or objective values once we obtain one of them.

Another celebrating result is that, in either algorithm, we construct the matrices row by row, while the order in which we construct such rows does not matter, as indicated below.

\begin{proposition}\label{reorder}
  Replacing $r_i$ with $r_{\sigma(i)}$ for an arbitrary permutation~$\sigma$ does not change the optimality of the objective value and~$(0,1)$-matrix obtained by Algorithm~\ref{MAlgo1} or~Algorithm~\ref{MAlgo2}.
\end{proposition}

We prove the above proposition in Appendix~\ref{PMPOMapdA8}. We further interpret it with the two operations~$\ominus$ and $\oplus$ as follows.

\begin{proposition}
  Consider $\bm {c}\in \mathbb{N}^n$, and~$\bm {r},\bm s \in \mathbb{N}^m$.
  \begin{enumerate}
    \item[$\mathit{1}$)] $\left(\bm c \oplus\bm{r}\right) \oplus \bm{s} = \left(\bm c \oplus\bm{s}\right) \oplus \bm{r}= \bm c \oplus [\bm r'~\bm s']'$.
    \item[$\mathit{2}$)] If $\bm {c}\prec^w [\bm r'~\bm s']'$, then \newline
    {$\left(\bm {c}\ominus\bm{r}\right)\ominus\bm{s}=\left(\bm {c}\ominus\bm{s}\right)\ominus\bm{r}=\bm {c}\ominus [\bm r'~\bm s']'.$}
  \end{enumerate}
\end{proposition}

%
Thus, we can respectively delineate the decompositions in the peak-shaving and valley-filling algorithms by
\begin{align*}
   \bm c\ominus \bm r=\bm c \ominus r_1\ominus \cdots \ominus r_m \text{ and }\bm b\oplus \bm r=\bm b \oplus r_1\oplus \cdots \oplus r_m.
\end{align*}
These decompositions give a fine-grained perspective on the inherent symmetry of the two iPOP problems. Now, let us use three toy examples to illustrate the decompositions.
\begin{figure}[t!]
  \mbox{\quad}\begin{tabular}{c|ccccc}
    \hline
    $\bm c$&$7$&$6$&$5$&$4$&$4$\\\hline\hline
    $\bm {\bar{c}}^{(0)}$&$7$&$6$&$5$&$4$&$4$\\\hline
    \end{tabular} $\Rightarrow$
   \begin{tabular}{c|ccccc}
    \hline
    $\bm {\bar{c}}^{(0)}$&$7$&$6$&$5$&$4$&$4$\\ \hline
    $r_1$&$1$&$1$&$1$&$1$&$0$\\\hline\hline
    $\bm {\bar{c}}^{(1)}$&$6$&$5$&$4$&$3$&$4$\\\hline
    \end{tabular}\smallskip   \\
    $\Rightarrow$\begin{tabular}{c|ccccc}
    \hline
    $\bm {\bar{c}}^{(1)}$&$6$&$5$&$4$&$3$&$4$ \\ \hline
    $r_2$&$1$&$1$&$1$&$0$&$1$
    \\\hline\hline
    $\bm {\bar{c}}^{(2)}$&$5$&$4$&$3$&$3$&$3$\\\hline
    \end{tabular} $\Rightarrow$
    \begin{tabular}{c|ccccc}
    \hline
    $\bm {\bar{c}}^{(2)}$&$5$&$4$&$3$&$3$&$3$ \\ \hline
    $r_3$&$1$&$1$&$1$&$0$&$0$
    \\\hline\hline
    $\bm {\bar{c}}^{(3)}$&$4$&$3$&$2$&$3$&$3$\\\hline
    \end{tabular}\smallskip  \smallskip  \\
    $\Rightarrow$
    \begin{tabular}{c|ccccc}
    \hline
    $\bm {\bar{c}}^{(3)}$&$4$&$3$&$2$&$3$&$3$ \\ \hline
    $r_4$&$1$&$0$&$0$&$0$&$0$
    \\\hline\hline
    $\bm {\bar{c}}^{(4)}$&$3$&$3$&$2$&$3$&$3$\\\hline
    \end{tabular} $\Rightarrow$\!
    \begin{tabular}{c|ccccc}
    \hline
    $\bm {\bar{c}}^{(4)}$&$3$&$3$&$2$&$3$&$3$\\ \hline
    $r_5$&$0$&$0$&$0$&$0$&$1$
    \\\hline\hline
    $\bm {\bar{c}}$&$3$&$3$&$2$&$3$&$2$\\\hline
    \end{tabular}
    \caption{The second loop of Algorithm~\ref{MAlgo1} in Example~\ref{Malgo1ex1}.}\label{figexMalgo1ex1}
    \end{figure}
\begin{figure}[t!]
   \mbox{\quad}\begin{tabular}{c|ccccc}
    \hline
    $\bm c$&$7$&$6$&$5$&$4$&$4$\\\hline\hline
    $\bm {\bar{c}}^{(0)}$&$7$&$6$&$5$&$4$&$4$\\\hline
    \end{tabular} $\Rightarrow$\!
   \begin{tabular}{c|ccccc}
    \hline
    $\bm {\bar{c}}^{(0)}$&$7$&$6$&$5$&$4$&$4$\\ \hline
    $r_5$&$1$&$0$&$0$&$0$&$0$\\\hline\hline
    $\bm {\bar{c}}^{(1)}$&$6$&$6$&$5$&$4$&$4$\\\hline
    \end{tabular}\smallskip   \\
    $\Rightarrow$ \begin{tabular}{c|ccccc}
    \hline
    $\bm {\bar{c}}^{(1)}$&$6$&$6$&$5$&$4$&$4$ \\ \hline
    $r_4$&$1$&$0$&$0$&$0$&$0$
    \\\hline\hline
    $\bm {\bar{c}}^{(2)}$&$5$&$6$&$5$&$4$&$4$\\\hline
    \end{tabular} $\Rightarrow$\!
    \begin{tabular}{c|ccccc}
    \hline
    $\bm {\bar{c}}^{(2)}$&$5$&$6$&$5$&$4$&$4$\\ \hline
    $r_3$&$1$&$1$&$1$&$0$&$0$
    \\\hline\hline
    $\bm {\bar{c}}^{(3)}$&$4$&$5$&$4$&$4$&$4$\\\hline
    \end{tabular} \smallskip   \\
    $\Rightarrow$ \begin{tabular}{c|ccccc}
    \hline
    $\bm {\bar{c}}^{(3)}$&$4$&$5$&$4$&$4$&$4$\\ \hline
    $r_2$&$1$&$1$&$1$&$1$&$0$
    \\\hline\hline
    $\bm {\bar{c}}^{(4)}$&$3$&$4$&$3$&$3$&$4$\\\hline
    \end{tabular} $\Rightarrow$\!
    \begin{tabular}{c|ccccc}
    \hline
    $\bm {\bar{c}}^{(4)}$&$3$&$4$&$3$&$3$&$4$\\ \hline
    $r_1$&$1$&$1$&$1$&$0$&$1$
    \\\hline\hline
    $\bm {\bar{c}}$&$2$&$3$&$2$&$3$&$3$\\\hline
    \end{tabular}\smallskip 
    \caption{The second loop of Algorithm~\ref{MAlgo1} in Example~\ref{Malgo1ex2}.}\label{figexMalgo1ex2}
    \end{figure}
\begin{example}\label{Malgo1ex1}
  Consider~$\bm c=[7~6~5~4~4]'$~and~$\bm r=[4~4~3~1~1]'$. We demonstrate Algorithm~\ref{MAlgo1} in Fig.~\ref{figexMalgo1ex1}. Moreover, the particular optimal objective value in Proposition~\ref{propcan} is~{$\bm c\ominus \bm r=[3~3~3~2~2]'$} and that in Proposition~\ref{rescanon} can be~$\bm c\ominus \bm r$, $[3~3~2~2~3]'$, $[3~2~3~2~3]'$, $[2~3~3~2~3]'$, or~$[2~3~2~3~3]'$.
  \end{example}
\begin{example}\label{Malgo1ex2}
  We illustrate Proposition~\ref{reorder} with the same~$(\bm c,\bm r)$ as in Example~\ref{Malgo1ex1}. Given a permutation reversing the original order of rows, we show the resulting peak-shaving process in Fig.~\ref{figexMalgo1ex2}. Note that the objective values generated in Example~\ref{Malgo1ex1} and this example are equivalent, agreeing with Proposition~\ref{reorder}.
\end{example}
\begin{example}\label{Malgo2ex1}
  Consider~$\bm b=[8~6~5~2~2]'$~and~$\bm r=[4~3~3~2~1]'$. We demonstrate Algorithm~\ref{MAlgo2} in Fig.~\ref{figexMalgo2ex1}. Moreover, the particular optimal objective value in Proposition~\ref{propcan} is~$\bm b\oplus \bm t=[8~8~7~7~6]'$ and that in Proposition~\ref{rescanon} can be~$[8~7~8~6~7]'$ or~$[8~8~7~6~7]'$.
 \end{example}

\begin{figure}[t!]
  \mbox{\quad}\begin{tabular}{c|ccccc}
    \hline
    $\bm b$&$8$&$6$&$5$&$2$&$2$\\\hline\hline
    $\bm {\bar{b}}^{(0)}$&$8$&$6$&$5$&$2$&$2$\\\hline
    \end{tabular} $\Rightarrow$\!
   \begin{tabular}{c|ccccc}
    \hline
    $\bm {\bar{b}}^{(0)}$&$8$&$6$&$5$&$2$&$2$\\ \hline
    $r_1$&$0$&$1$&$1$&$1$&$1$\\\hline\hline
    $\bm {\bar{b}}^{(1)}$&$8$&$7$&$6$&$3$&$3$\\\hline
    \end{tabular}\smallskip \\
    $\Rightarrow$
    \begin{tabular}{c|ccccc}
    \hline
    $\bm {\bar{b}}^{(1)}$&$8$&$7$&$6$&$3$&$3$ \\ \hline
    $r_2$&$0$&$0$&$1$&$1$&$1$
    \\\hline\hline
    $\bm {\bar{b}}^{(2)}$&$8$&$7$&$7$&$4$&$4$\\\hline
    \end{tabular} $\Rightarrow$\!
    \begin{tabular}{c|ccccc}
    \hline
    $\bm {\bar{b}}^{(2)}$&$8$&$7$&$7$&$4$&$4$ \\ \hline
    $r_3$&$0$&$0$&$1$&$1$&$1$
    \\\hline\hline
    $\bm {\bar{b}}^{(3)}$&$8$&$7$&$8$&$5$&$5$\\\hline
    \end{tabular} \smallskip   \\
    $\Rightarrow$
    \begin{tabular}{c|ccccc}
    \hline
    $\bm {\bar{b}}^{(3)}$&$8$&$7$&$8$&$5$&$5$ \\ \hline
    $r_4$&$0$&$0$&$0$&$1$&$1$
    \\\hline\hline
    $\bm {\bar{b}}^{(4)}$&$8$&$7$&$8$&$6$&$6$\\\hline
    \end{tabular} $\Rightarrow$\!
    \begin{tabular}{c|ccccc}
    \hline
    $\bm {\bar{b}}^{(4)}$&$8$&$7$&$8$&$6$&$6$\\ \hline
    $r_5$&$0$&$0$&$0$&$1$&$0$
    \\\hline\hline
    $\bm {\bar{b}}$&$8$&$7$&$8$&$7$&$6$\\\hline
    \end{tabular} 
    \caption{The second loop of Algorithm~\ref{MAlgo2} in Example~\ref{Malgo2ex1}.}\label{figexMalgo2ex1}
    \end{figure}

As a by-product, the two algorithms uncover that the extra elemenwise inequality in Problem~\eqref{basicmaj1} does not complicate the problem solving. Given this fact, we propose more generalized iPOP problems that can also be solved by our peak-shaving or valleying filling approach in Appendix~\ref{GiPOP}. Furthermore, our results regarding these problems extend several existing results, e.g., Theorem~1 in~\cite{kim1998simple} and Lemma~1 in~\cite{mo2022Tensor}.

\begin{figure*}[t]
 \centering
  \subfloat[Comparison (row number).]{\label{MAlgo1a}\includegraphics[width=0.96\columnwidth, height=0.64\columnwidth]{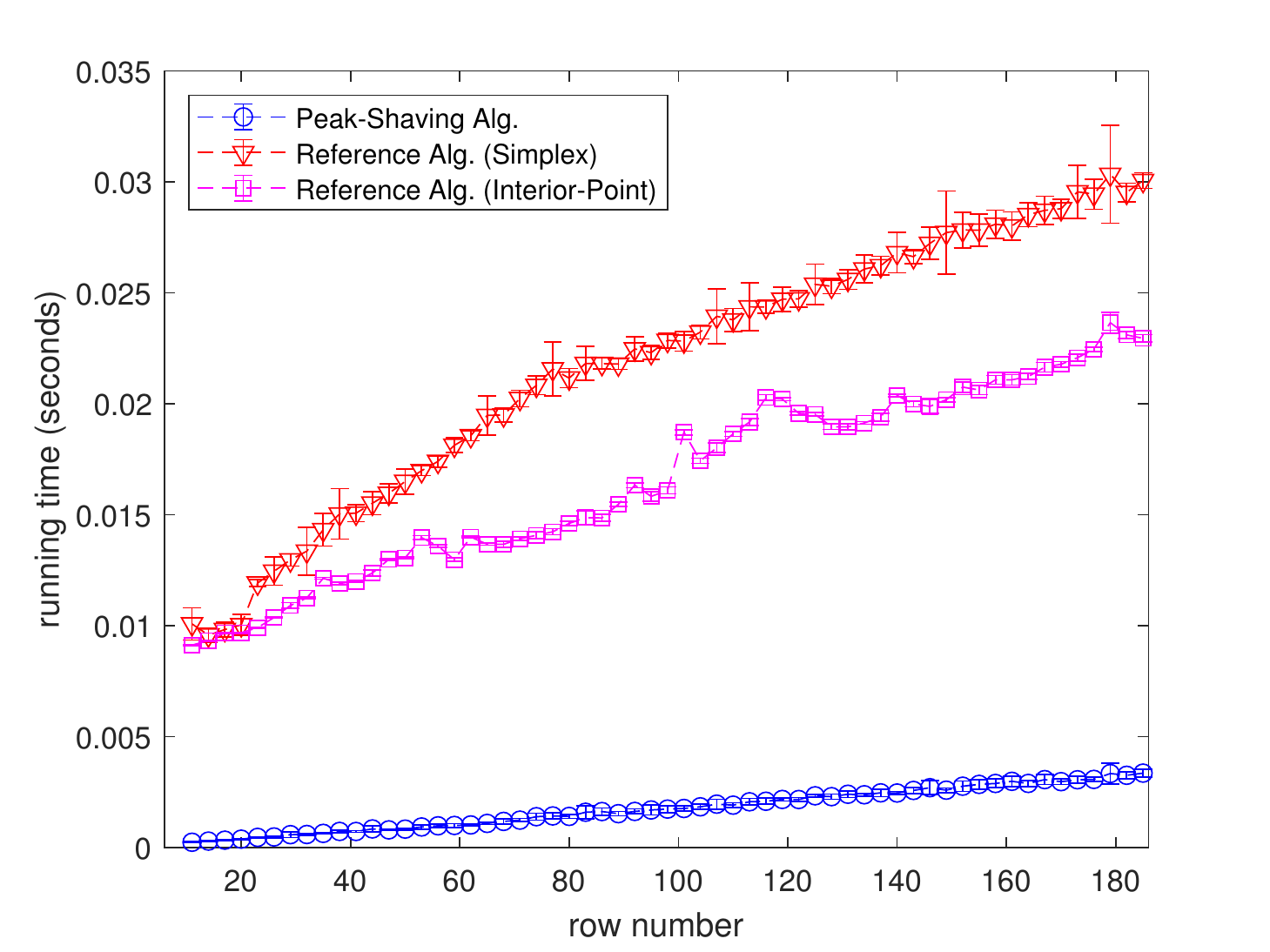}}
  \subfloat[Fitting (row number).]{\label{MAlgo1b}\includegraphics[width=0.96\columnwidth, height=0.64\columnwidth]{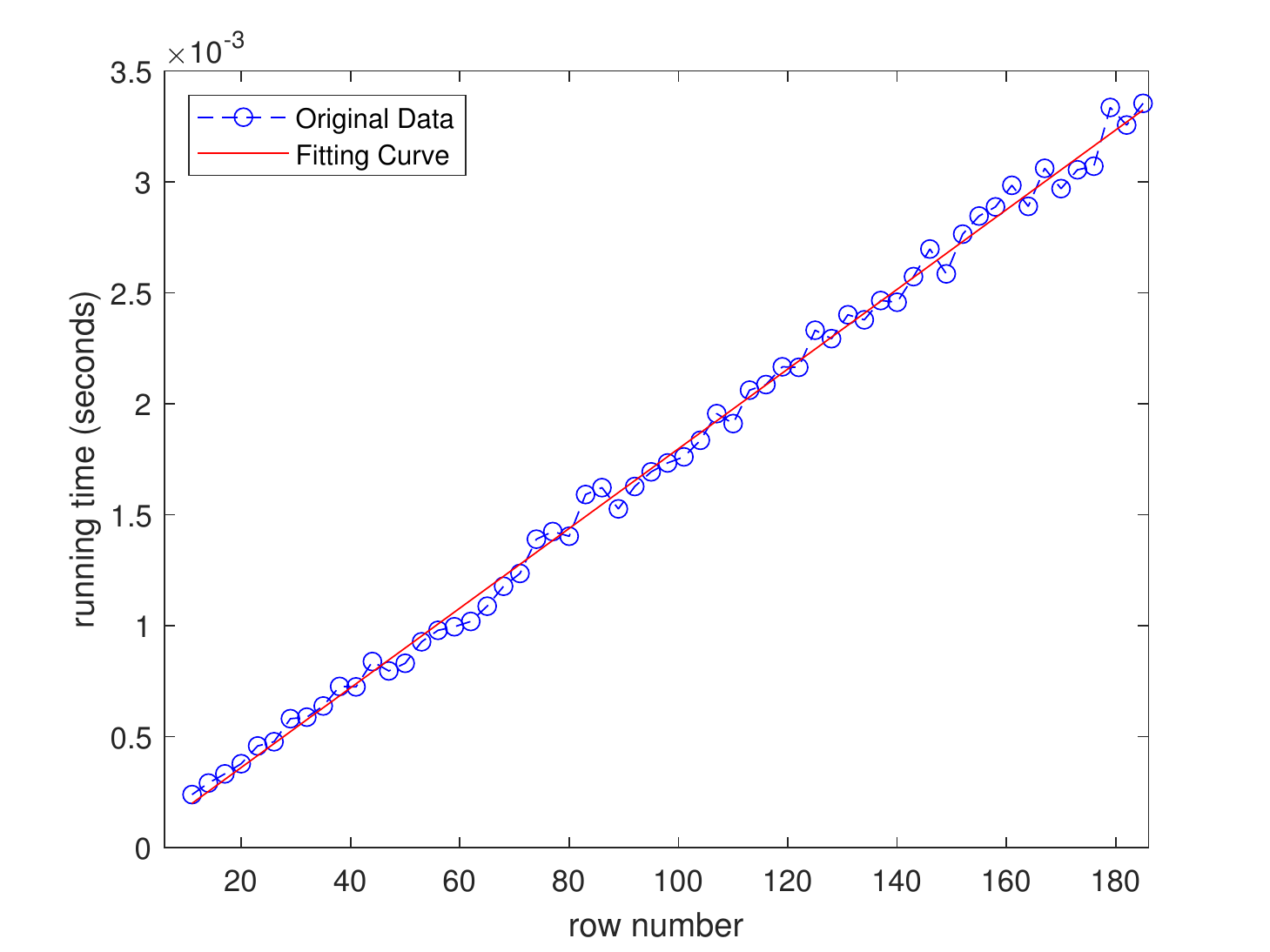}}\\
  \subfloat[Comparison (column number).]{\label{MAlgo1c}\includegraphics[width=0.96\columnwidth, height=0.64\columnwidth]{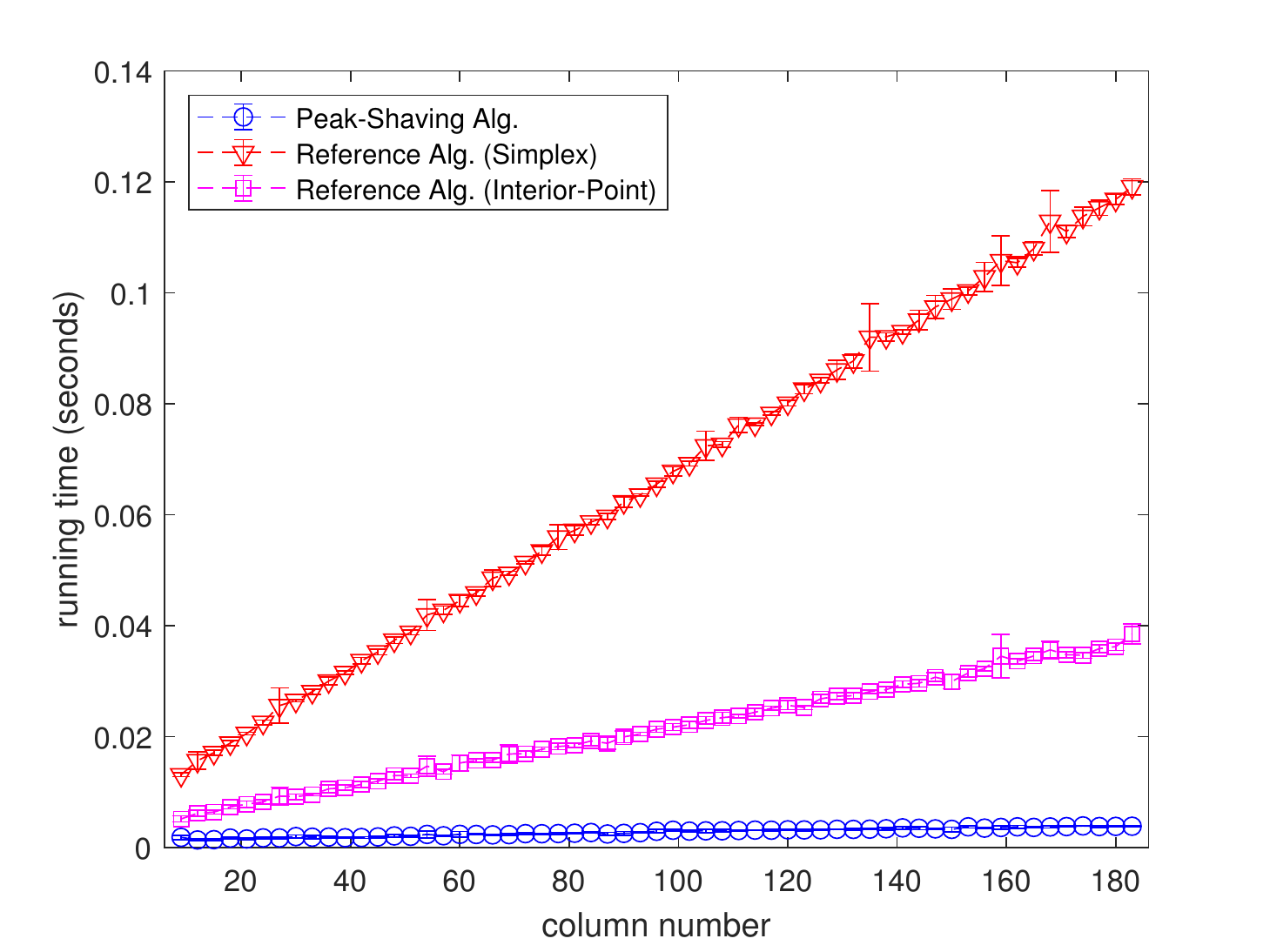}}
  \subfloat[Fitting (column number).]{\label{MAlgo1d}\includegraphics[width=0.96\columnwidth, height=0.64\columnwidth]{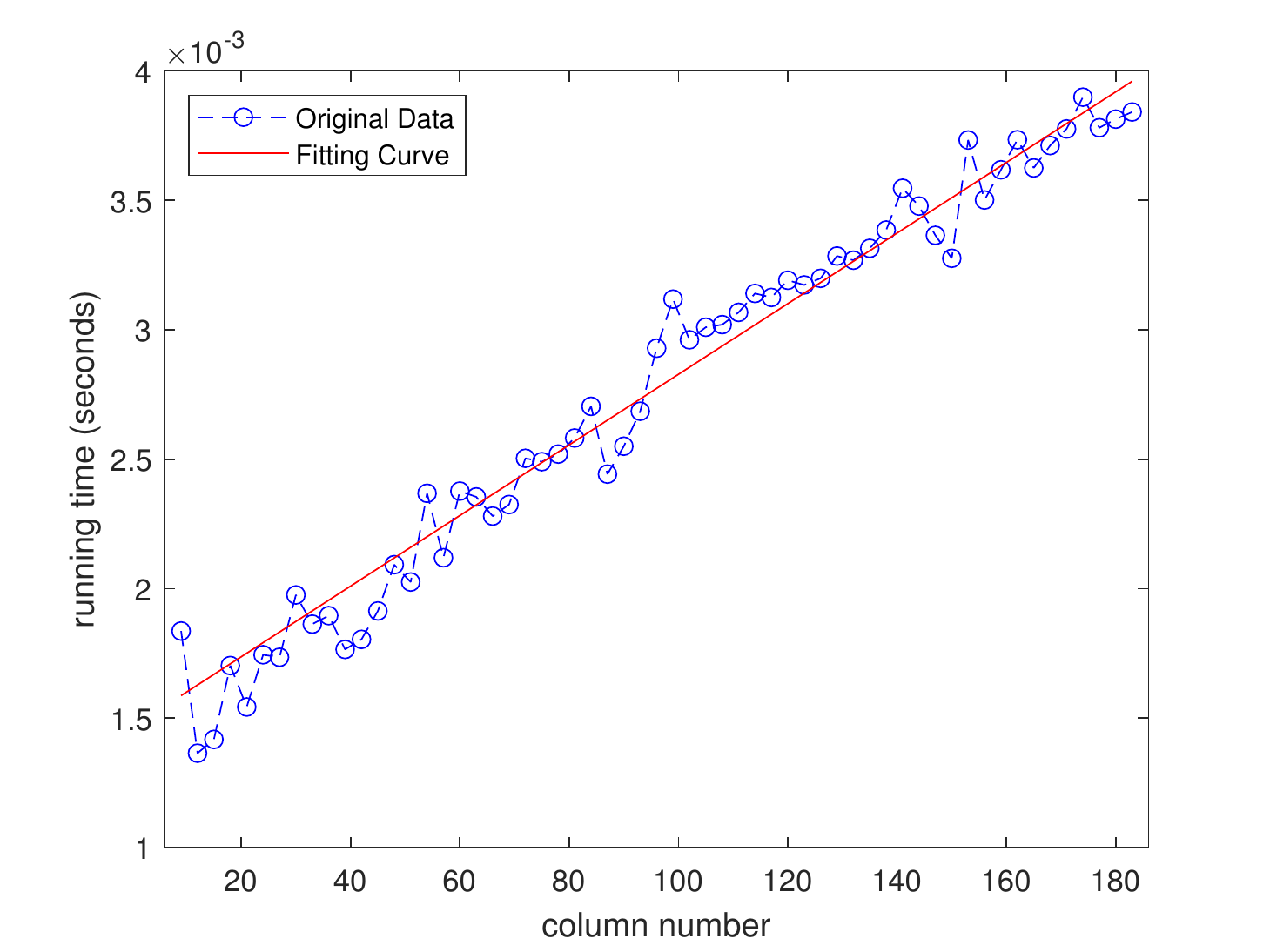}}
  \caption{Numerical simulation regarding Algorithm~\ref{MAlgo1}~--~Peak-shaving.}\label{MAlgo1all}
\end{figure*}

\section{Comparative Analysis and Generalization}\label{sec_comparison}
In this section, we shall show the efficiency of our ``peak-shaving'' or ``valley-filling'' approach in comparison with the traditional order-preserving scalarization method for the proposed iPOP problems. To this end, we also present natural extensions of our iPOP problems, further verifying the strength of our solution method. 

\begin{figure*}[t]
\centering
  \subfloat[Comparison (row number).]{\label{MAlgo2a}\includegraphics[width=0.96\columnwidth, height=0.66\columnwidth]{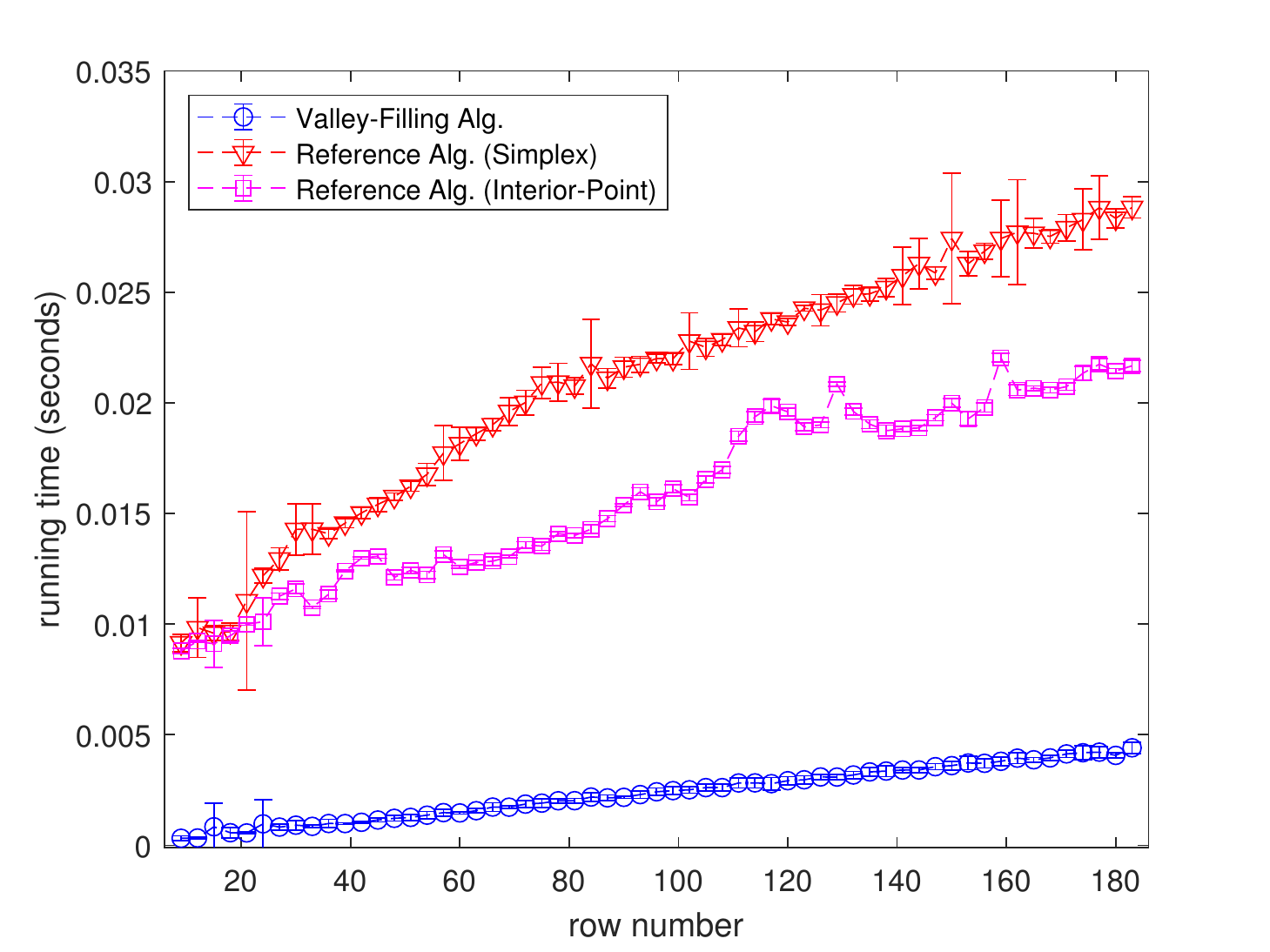}}
  \subfloat[Fitting (row number).]{\label{MAlgo2b}\includegraphics[width=0.96\columnwidth, height=0.66\columnwidth]{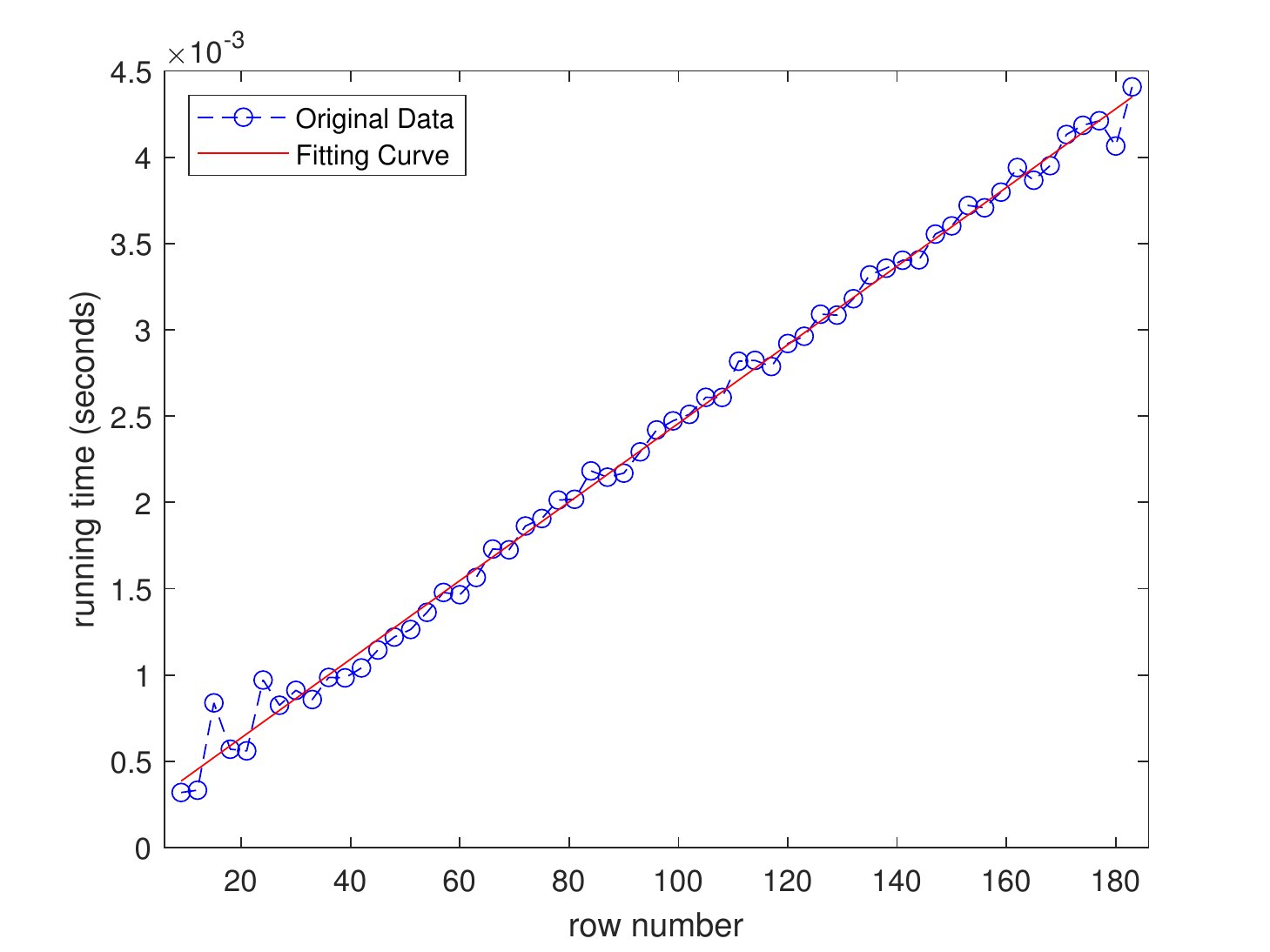}}\\
  \subfloat[Comparison (column number).]{\label{MAlgo2c}\includegraphics[width=0.96\columnwidth, height=0.66\columnwidth]{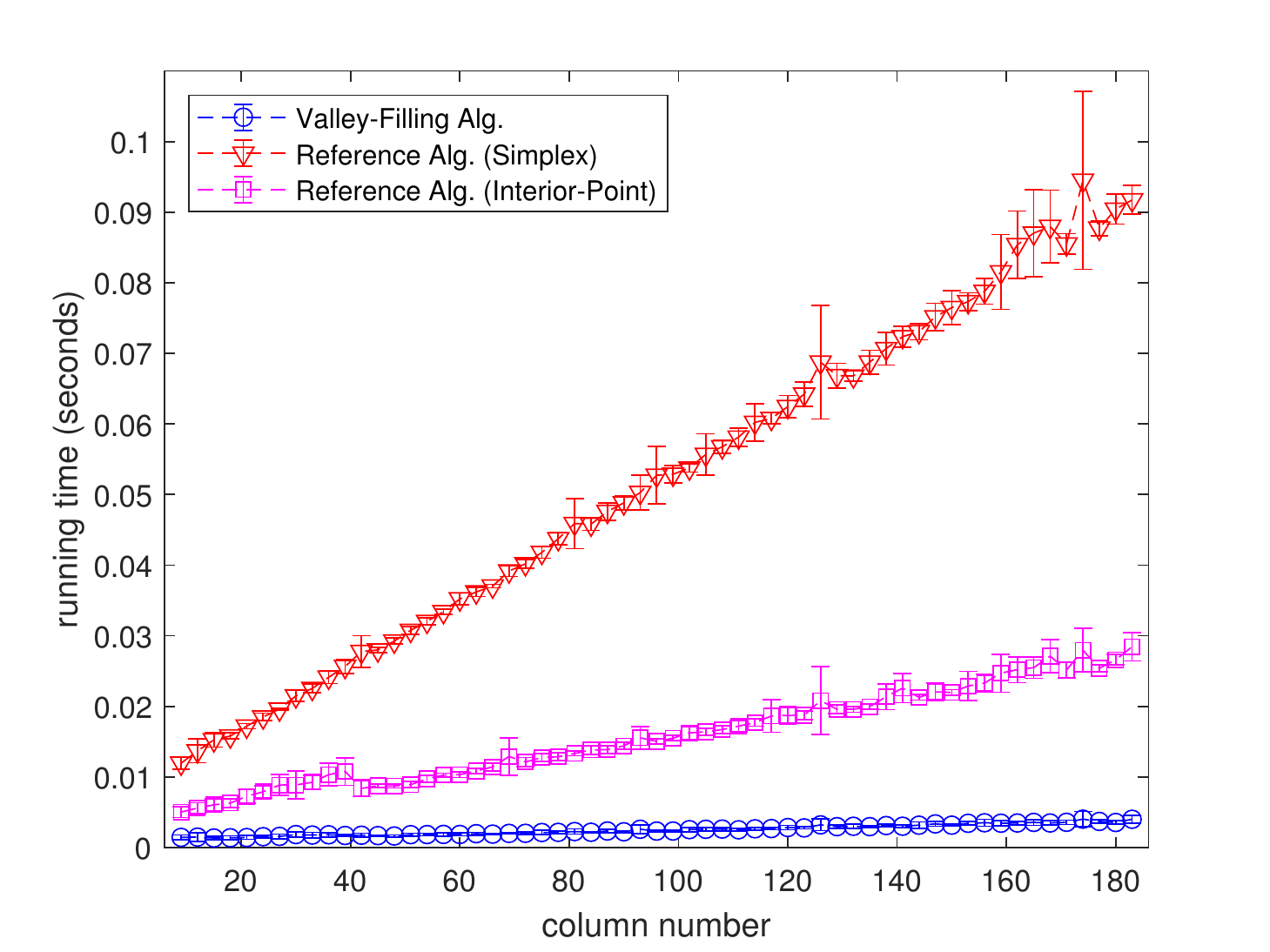}}
  \subfloat[Fitting (column number).]{\label{MAlgo2d}\includegraphics[width=0.96\columnwidth, height=0.66\columnwidth]{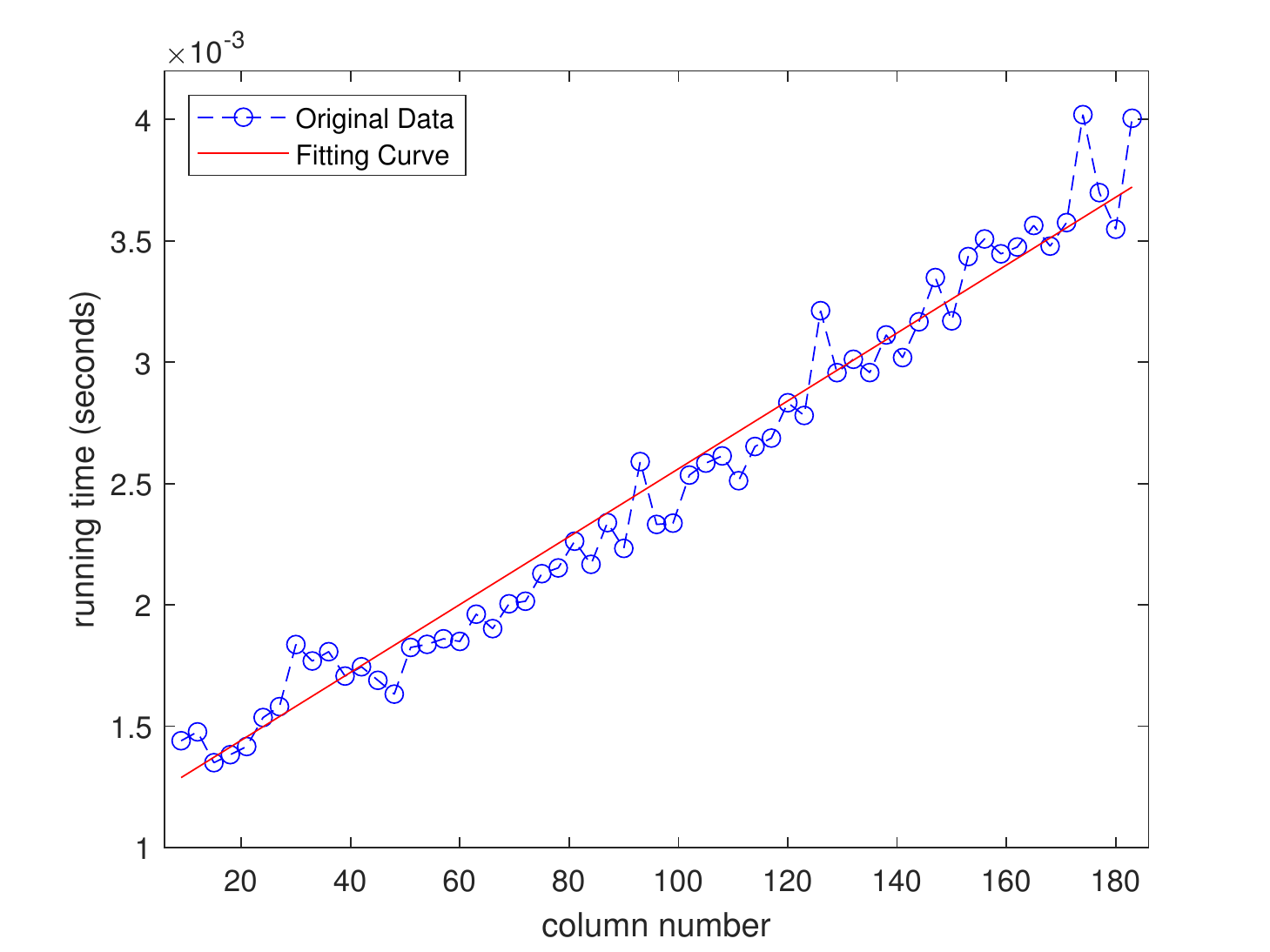}}
  \caption{Numerical simulation regarding Algorithm~\ref{MAlgo2}~--~Valley-filling.}\label{MAlgo2all}
\end{figure*}

\subsection{Order-Preserving Scalarization}
As mentioned in the Introduction, the direct use of B\&B may not work well for our iPOP problems. Thus, we take the order-preserving scalarization for POP as a benchmark. Real-valued functions are said to be order-preserving if they convert partially ordered elements into real numbers without changing the original orders. For the majorization order, such functions are called Schur-convex functions.

Although the problems after scalarization are integer programs, we can exactly solve them by a sequence of associated linear programs~\cite{meyer1977class}. The reason lies in that the objective function is separable convex~(i.e., the sum of convex functions over independent variables) and the constraints are linear, accompanied by a totally unimodular constraint matrix~\cite[Section~19]{schrijver1998theory}. Specifically, we first introduce an auxiliary variable~$\bm p\in \mathbb{N}^n$ by~$\bm p+\sum_{i=1}^{m}\bm a_{i\bigcdot}'=\bm c$ and apply the sum-of-squares function that is strictly Schur-convex to the optimization objective. Then, we obtain a tractable integer program:

\noindent \begin{equation}\label{basicmajintp}
\begin{split}
\min_{\bm p, A} &~~~~\sum\nolimits_{j=1}^{n}p_j^2\ \\
\text{subject to~}& ~~~~ A\in \{0,1\}^{m\times n};\;\bm p\in \mathbb{N}^n;\; \\
& ~~~~\|\bm a_{i\bigcdot}\|_1=r_i\text{, for all }i\in\underline{m};~~\\
& ~~~~\|\bm a_{\bigcdot j}\|_1+p_j=c_j\text{, for all }j\in \underline{n}.
\end{split}
\end{equation}
How to solve Problem~(\ref{basicmajintp}) is out of our scope, and interested readers can refer to~\cite{meyer1977class}. Note that it is harder to solve Problem~(\ref{basicmajintp}) than a linear program with~$(mn+n)$ variables.

\subsection{Comparison and Simulation}
Next, we detailedly compare our specialized approach with the common order-preserving scalarization method. First, our approach reveals the essential uniqueness of all optimal objective values of each iPOP problem (see Theorem~\ref{leastthm}). However, we can hardly draw such a conclusion from the scalarization method. Second, Algorithms~\ref{MAlgo1} and~\ref{MAlgo2} agree with the rules of thumb ``peak shaving'' and ``valley filling,'' respectively. The two iPOP problems require~$\bm c-\bm {x}$ or~$\bm b+\bm {x}$ be as flat as possible by selecting~$\bm {x}$. Our approach uncovers such a fact in both the final status of~$\bm c-\bm {x}$ or~$\bm b+\bm {x}$ and constructing each row of the associated optimal matrices. Third, our approach indicates the minor role of the elementwise inequality constraint~$\bm {x}\leq \bm c$ in finding an optimal solution, as shown in Algorithm~\ref{MAlgo1}. This fact justifies the symmetry of Problem~(\ref{basicmajmatrix1}) and Problem~(\ref{basicmajmatrix2}); moreover, it motivates us to study more generalized iPOP problems in Section~\ref{GiPOP}, extending certain results in~\cite{mo2022Tensor} and~\cite{kim1998simple}. Again, it is hard to make the symmetry characteristic and generalization clear by using the scalarization method. Overall, from a theoretical perspective, we are in favor of our approach to Problems~(\ref{basicmajmatrix1}) and~(\ref{basicmajmatrix2}), since the former reveals the essential properties of the proposed iPOP problems and agrees with intuitive concepts of ``peak shaving'' and ``valley filling.''

From the computational perspective, it is easier to implement Algorithms~\ref{MAlgo1} and~\ref{MAlgo2} than the scalarization approach since our algorithms involve simple operations like sorting and subtraction/addition. This makes us suffer less from the curse of dimensionality. In addition, the scalarization approach involves real numbers, but our approach only deals with integers, avoiding round-off errors. Recall that Algorithms~\ref{MAlgo1} and~\ref{MAlgo2} both have the linear time complexity~$\mathcal{O}(mn)$. In contrast to this, the scalarization approach is much harder than solving a linear program with~$(mn+n)$ variables, and there is no linear programming solver that possesses a time complexity linear to the number of variables.

Bearing these theoretical strengths in mind, we further corroborate the computational efficacy by numerical simulations. Specifically, we compare the running time of Algorithms~\ref{MAlgo1} and~\ref{MAlgo2} with that of characterizing the feasibility of the relaxed versions of the scalarized problems, respectively, via MATLAB. We call the latter the reference algorithm, which is just one of the many steps in the scalarization method. Given setups~$(\bm c, \bm r)$ and~$(\bm b, \bm r)$, note that the column number~$n$ refers to the length of~$\bm b$ or~$\bm c$, while the row number~$m$ refers to that of~$\bm r$ in the figures. The remaining data are randomly generated. 

We display the simulation results on Algorithm~\ref{MAlgo1} in Fig.~\ref{MAlgo1all}. In Fig.~\ref{MAlgo1a} and~Fig.~\ref{MAlgo1c}, we compare the peak-shaving algorithm with the reference algorithm. The two figures indicate that the former is much more efficient than the latter under the simplex and interior-point methods for different row~or~column numbers, because the former has shorter average running times and standard deviations. Moreover, by using a linear function to fit the relationship between the row number and the running time, we observe in Fig.~\ref{MAlgo1b} that we can use a linear function to properly fit the data labeled by circles in Fig.~\ref{MAlgo1a}. This observation means that the running time of Algorithm~\ref{MAlgo1} grows linearly as the row number increases. Similarly, we show in Fig.~\ref{MAlgo1d} that we can use a linear model to fit the data labeled by circles in Fig.~\ref{MAlgo1c}, meaning that the running time of Algorithm~\ref{MAlgo1} grows linearly as the column number increases. Such phenomena are consistent with Theorem~\ref{algthm}.

We present the simulation results of Algorithm~\ref{MAlgo2} in Fig.~\ref{MAlgo2all}. In Fig.~\ref{MAlgo2a} and~Fig.~\ref{MAlgo2c}, we compare the valley-filling algorithm with the reference algorithm. The two figures show that the former is much more efficient than the latter under both the simplex and interior-point methods, for different row~or~column numbers, because the former has shorter average running times and standard deviations. In Fig.~\ref{MAlgo2b}, we see that the data labeled by circles in Fig.~\ref{MAlgo2a} can be fitted by a linear function. This observation means that the running time of Algorithm~\ref{MAlgo2} grows linearly as the row number increases. Similarly, in Fig.~\ref{MAlgo2d}, we see that the data labeled by circles in Fig.~\ref{MAlgo2c} can be fitted by a linear function. This means that the running time of Algorithm~\ref{MAlgo2} grows linearly as the column number increases. Such phenomena are also consistent with Theorem~\ref{algthm}.

\subsection{Generalized iPOP problems} \label{GiPOP}
Given the inherent symmetry revealed by our approach, it is natural to examine two more generalized iPOP problems:
\begin{equation}
\begin{split}\label{basicmajmatrix3}
\underset{A}{\text{~~minimize}_{\prec}} &~~~~\bm d-\sum_{i=1}^{m}\bm a_{i\bigcdot}'\ \\
\text{subject to~}&~~~~A\in \{0,1\}^{m\times n};\;\|\bm a_{i\bigcdot}\|_1=r_i, \forall i\in\underline{m};\\
&~~~~ \|\bm a_{\bigcdot j}\|_1\leq c_j, \forall j\in \underline{n}. 
\end{split}
\end{equation}


\noindent
\begin{equation}
\begin{split}\label{basicmajmatrix4}
\underset{A}{\text{~~minimize}_{\prec}} &~~~~\bm b+\sum_{i=1}^{m}\bm a_{i\bigcdot}'\ \\
\text{subject to~}&~~~~A\in \{0,1\}^{m\times n};\;\|\bm a_{i\bigcdot}\|_1=r_i, \forall i\in\underline{m};\\
&~~~~ \|\bm a_{\bigcdot j}\|_1\leq c_j, \forall j\in \underline{n}. 
\end{split}
\end{equation}
where~$\bm a,\bm b, \bm c \in \mathbb{N}^n$ and~$\bm r\in\mathbb{N}^m$. Clearly, Problem~(\ref{basicmajmatrix3}) extends Problem~(\ref{basicmajmatrix1}) by allowing~$\bm c\neq \bm d$ , and Problem~(\ref{basicmajmatrix4}) augments Problem~(\ref{basicmajmatrix2}) by adding inequality constraints~$\|\bm a_{\bigcdot j}\|_1\leq c_j$, for all~$j\in \underline{n}$. Meanwhile, Problem~(\ref{basicmajmatrix3}) and Problem~(\ref{basicmajmatrix4}) coincide when~\mbox{$\bm b=-\bm d$}, because~$\bm x\prec \bm y$ if and only if~$-\bm x\prec -\bm y$.

According to the two developed optimal matrix-completion algorithms, we find out that the elementwise inequality constraint does not complicate the problem solving. Thus, we can conclude that our peak-shaving or valley-filling approach can also help solve Problem~(\ref{basicmajmatrix3}) and Problem~(\ref{basicmajmatrix4}). The corresponding algorithms almost remain the same with Algorithm~\ref{MAlgo1} and Algorithm~\ref{MAlgo2}. Specifically, we still follow the rule of thumb ``peak shaving'' or ``valley filling'' to sequentially construct the rows. The difference lies in that if the column sum of the partial~$(0,1)$-matrix reaches its bound specified by an element in~$c$, then we will not assign one to the intersection between this column and every future row such that the elementwise inequality constraint is satisfied. Moreover, We can use similar techniques to attain an optimal objective value if we change the objective from ``$\text{minimize}_{\prec}$'' to ``$\text{maximize}_{\prec}$'' for the maximal elements of the attainable objective values.


Furthermore, Problem~(\ref{basicmajmatrix4}) with~$\bm b=\bm 0$ and~$\bm r=\bm 1$ is a discrete version of the problem in Theorem~2 of~\cite{kim1998simple}. This fact will be more clear if we rewrite Problem~(\ref{basicmajmatrix4}) in its equivalent compact formulation. 
Meanwhile, when our objective is to maximize~$\bm x$ in terms of majorization, the problem relates to Theorem~1 in~\cite{kim1998simple}. From this perspective, our results also generalize the study of~\cite{kim1998simple} on load assignment in crossbar switches with output queuing. Last but not least, our results on Problem~(\ref{basicmajmatrix4}) cover Lemma~1 in~\cite{mo2022Tensor}, where we restrict~$\bm r=\bm 1$. So far, we have shown that our iPOP problems extended existing studies by allowing general base, ceiling, and row sum vectors. These extensions also verify the advantages of our approach.
\section{Conclusions}\label{sec_conclusion}
In this paper, we propose and study two optimal~$(0,1)$-matrix completion problems with majorization ordered objectives, whose compact formulations are optimization problems over majorization-ordered lattices. We show their applications in electric vehicle charging, portfolio optimization, and secure data storage. Interestingly, we prove that the optimal objective values of each proposed iPOP problem are essentially unique in the sense that they share the same nonincreasing rearrangement. We also identify two particular optimal objective values characterized  by the order of the elements in the objective value or the corresponding column sum vector. Furthermore, we respectively propose a peak-shaving and valley-filling approach to obtain all optimal objective values and the associated optimal solutions for the two iPOP problems. Theoretical analysis and numerical simulations corroborate the strengths of our approach compared to the standard order-preservation scalarization method for POP.

In the future, a follow-up work is to advance the unstructured matrix completion with POP in this work to structured ones of staircase, banded, or even arbitrary patterns~\cite{fulkerson1960zero,chen1966realization,anstee1982properties,chen2016constrained,mo2020market,mo2022Tensor}. Also, we shall apply POP to more application scenarios that bring us more insightful analysis and efficient solution approaches. Moreover, it is interesting to study POP in an online fashion of real-time uncertainty.

\appendix\label{PMPOMapdA}
\subsection{Useful Definitions and Lemmas} \label{PMPOMapdA0}
The following lemma comes from~\cite[Section~5.A]{olkin2016inequalities} and will be useful in the proofs of our key results.
\begin{lemma}
 For $\bm x, \bm y \in \mathbb{R}^n$, the following are equivalent:
      \begin{enumerate}
        \item[$\mathit{1}$)] $\bm x \prec^w \bm y$;
        \item[$\mathit{2}$)] There exists~$\bm u\in \mathbb{R}^n$ with $\bm u \prec \bm y$ and $\bm u \leq \bm x$;
        \item[$\mathit{3}$)] There exists~$\bm v\in \mathbb{R}^n$ with $\bm x \prec \bm v$ and~$\bm y \leq \bm v$.
      \end{enumerate}
 \label{realfeacondcombineprecminusmaj}
\end{lemma}

Next, we define an operation that will be helpful in proving a critical lemma.
\begin{definition}
  An adjacent swap of a vector~$\bm x$ is an operation which exchanges two adjacent elements of~$\bm x$. An adjacent swap is said to be positive if it brings the bigger element of the two exchanged elements forward; otherwise, it is negative.
\end{definition}

The nonincreasing rearrangement of a vector can be obtained from the vector by a finite sequence of positive adjacent swaps, while the nondecreasing rearrangement can be obtained from it via a sequence of negative adjacent swaps. Such swap sequences for both cases may not necessarily be unique. With such observations, we are ready to state the following lemmas.

\begin{lemma}
    For $\bm x,\bm y\in \mathbb{R}^n$, if $\bm y\leq \bm x$, then $\bm y^\downarrow \leq \bm x^\downarrow$, and equivalently $\bm y^\uparrow\leq \bm x^\uparrow$. \label{comadjswap}
\end{lemma}
\begin{proof}
   We start with a simple case where~$n=2$ and assume~$\bm x=\bm x^\downarrow$. If~$\bm y=\bm y^\downarrow$, then $\bm y^\downarrow\leq \bm x^\downarrow$ clearly. Otherwise,~$\bm y=\begin{bmatrix}y_{[2]}&y_{[1]}\end{bmatrix}'$. By~$\bm y\leq \bm x$, we have~$y_{[1]}\leq x_{[2]}\leq x_{[1]}$ and $y_{[2]}\leq y_{[1]}\leq x_{[2]}$. Thus, we prove the case with~$n=2$.

  For a general~$n$, we again assume that $\bm x=\bm x^\downarrow$. Select an arbitrary sequence of positive adjacent swaps which transforms~$\bm y$ into~$\bm y^\downarrow$. Applying the previous result for~$n=2$ to each adjacent swap in the process from $\bm y$ to $\bm y^\downarrow$, we see that the resultant rearrangement of~$\bm y$ is always no greater than $\bm x$. Therefore, we have $\bm x^\downarrow \leq \bm y^\downarrow$. It follows that $\bm x^\uparrow\leq \bm y^\uparrow$. 
\end{proof}

Lemma~\ref{comadjswap} states that sorting will not change the elementwise inequality between two vectors. The following lemma further gives two majorization inequalities on sorting. The proof is similar to that of Lemma~\ref{comadjswap} and thus omitted for brevity.

\begin{lemma}\label{comadjswap2}
    For $\bm x,\bm y\in \mathbb{R}^n$, the following inequalities hold:
    $$\bm x^\downarrow-\bm y^\downarrow\prec \bm x^\downarrow-\bm y\text{ and }\bm x^\downarrow+\bm y^\uparrow\prec \bm x^\downarrow+\bm y.$$
\end{lemma}

The following lemma will be a critical component to prove our key results.
\begin{lemma} \label{deducadditionseq}
  Consider nonnegative integer vectors~\mbox{$\bm x, \bm y\in \mathbb{N}^n$} and~two sequences of indices, namely,~$(p_1,p_2,\ldots, p_\tau)$ and~$(q_1,q_2,\ldots,q_\tau)$, such that $p_i,q_i\leq n$ for~$i\in \underline{\tau}$ and the indices in the same sequence are distinct from each other.
  \begin{enumerate}
    \item[$\mathit{1}$)] If~$\bm x \prec \bm y$ and~$p_1\leq q_1, p_2\leq q_2, \dots, p_\tau\leq q_\tau$, then~$$\bm w \prec \bm z\text{, where }\bm w=\bm x^\downarrow- \sum_{l=1}^{\tau}\bm{e_{p_l}}\text{ and }\bm z=\bm y^\downarrow- \sum_{l=1}^{\tau}\bm{e_{q_l}}.$$
    \item[$\mathit{2}$)] If~$\bm x \prec \bm y$ and~$p_1\geq q_1, p_2\geq q_2, \dots, p_\tau\geq q_\tau$, then $$\bm w \prec \bm z\text{, where }\bm w=\bm x^\downarrow+ \sum_{l=1}^{\tau}\bm{e_{p_l}}\text{ and } \bm z=\bm y^\downarrow+\sum_{l=1}^{\tau}\bm{e_{q_l}}.$$
  \end{enumerate}
\end{lemma}
\begin{proof}
  We start with a simple case of the first part, namely, for~$\bm x, \bm y\in \mathbb{N}^n$ and~$p,q\in \mathbb{N}$, if $\bm x\prec \bm y$ and $1\leq p \leq q\leq n$, then it holds that~$\bm x^\downarrow - \bm{e_p}\prec \bm y^\downarrow- \bm{e_q}$. To show this, we let $\hat{p}$ denote the largest index such that $x_{[p]}=x_{\hat{p}}$. Similarly, let~$\hat{q}$ denote the largest index such that $y_{[q]}=y_{\hat{q}}$. Clearly,~$\bm x^\downarrow - \bm{e_p}\sim\bm x^\downarrow - \bm{e_{\hat{p}}}$ and $\bm y^\downarrow - \bm{e_q}\sim\bm y^\downarrow - \bm{e_{\hat{q}}}$. Thus, to show~$\bm x^\downarrow - \bm{e_p}\prec \bm y^\downarrow- \bm{e_q}$, it is enough to show that~$\bm x^\downarrow - \bm{e_{\hat{p}}}\prec \bm y^\downarrow - \bm{e_{\hat{q}}}$ when $\bm x\prec \bm y$ and~$p \leq q$. For brevity, we use $\bm w$ and $\bm z$ to denote $\bm x^\downarrow - \bm{e_{\hat{p}}}$ and~$\bm y^\downarrow - \bm{e_{\hat{q}}}$, respectively. Thus, we have~$\bm w=\bm w^\downarrow$ and $\bm z= \bm z^\downarrow$.

  For each $k<\min\{\hat{p},\hat{q}\}$ or $k\geq \max\{\hat{p},\hat{q}\}$, we have
  \begin{align*}
     & \sum\limits_{l=1}^k w_l=\sum\limits_{l=1}^k x_{[l]}\leq \sum\limits_{l=1}^k y_{[l]}=\sum\limits_{l=1}^k z_{l} \text{ and}\\
    & \sum\limits_{l=1}^k w_l=\sum\limits_{l=1}^k x_{[l]}-1\leq \sum\limits_{l=1}^k y_{[l]}-1=\sum\limits_{l=1}^k z_{l}.
  \end{align*}

  Moreover, the total sums of $\bm w$ and $\bm z$ are equal, since {\smallskip\begin{center}$\sum\limits_{l=1}^n w_l=\sum\limits_{l=1}^n x_{[l]}-1= \sum\limits_{l=1}^n y_{[l]}-1=\sum\limits_{l=1}^n z_{l}.$
  \end{center}\smallskip}

  If $\hat{p}=\hat{q}$, then~$\bm w\prec \bm z$.
  If $\hat{q}>\hat{p}$, for $\hat{p}\leq k<\hat{q}$, we have {\smallskip\begin{center}$\sum\limits_{l=1}^k w_l=\sum_{l=1}^k x_{[l]}-1< \sum\limits_{l=1}^k y_{[l]}=\sum\limits_{l=1}^k z_{l}.$
  \end{center}}

  If $\hat{p}>\hat{q}$, then we let~$\hat{k}$ be the smallest index such that~$\sum_{l=1}^{\hat{k}}w_l > \sum_{l=1}^{\hat{k}}z_l$, then we have~$p\leq q \leq \hat{q}\leq \hat{k}<\hat{p}$. As~$\sum_{l=1}^{\hat{k}}z_l=\sum_{l=1}^{\hat{k}}y_{[l]}-1\geq \sum_{l=1}^{\hat{k}}x_{[l]}-1= \sum_{l=1}^{\hat{k}}w_l-1$, it follows that
 \begin{equation*}
 \begin{split}
&\sum\limits_{l=1}^{\hat{k}}x_l = \sum\limits_{l=1}^{\hat{k}}y_l\text{ and }\\
  &  x_{[\hat{p}]}=x_{[\hat{p}-1]}=\cdots=x_{[\hat{k}]}\geq y_{[\hat{k}]}\geq y_{[\hat{k}+1]}\geq \cdots \geq y_{[\hat{p}]}.\end{split}
  \end{equation*}

  Together with the assumption $\sum_{l=1}^{\hat{k}}w_l > \sum_{l=1}^{\hat{k}}z_l$, we achieve $\sum_{l=1}^{\hat{p}}w_l > \sum_{l=1}^{\hat{p}}z_l$, which is a contradiction. Thus, such a $\hat{k}$ never exists, which completes the proof of the first part with~$\tau=1$. Repeatedly applying this result, we can easily prove the first part with a general~$\tau$. Moreover, the second part can be proven similarly.
\end{proof}


\subsection{Proof of Proposition~\ref{propcan}} \label{PMPOMapdA4}
   If $\bm v \in \mathcal{V}^\ominus$, then $\bm v \leq \bm c=\bm c^\downarrow $ and~$\bm c - \bm v \prec \bm t$. By Lemma~\ref{comadjswap} and Lemma~\ref{comadjswap2}, we show that~$\bm v^\downarrow \leq \bm c$ and~$\bm c - \bm v^\downarrow \prec \bm c - \bm v$, which implies~$\bm v^\downarrow \in \mathcal{V}^\ominus$. Next, if $\bm v \in \mathcal{V}^\oplus$, then~\mbox{$\bm v \geq \bm b=\bm b^\downarrow$} and~\mbox{$\bm v - \bm b \prec \bm t$}. It follows from~Lemma~\ref{comadjswap} that~$\bm v^\downarrow \geq \bm b$ and from~\cite[Section 1.A]{olkin2016inequalities} that~$\bm b - \bm v \prec -\bm t$. Moreover, by Lemma~\ref{comadjswap2}, we have~$\bm b - \bm v^\downarrow \prec \bm b - \bm v \prec -\bm t$. For similar arguments, we have~$\bm v^\downarrow - \bm b \prec \bm t$. Thus, we conclude that~$\bm v^\downarrow \in \mathcal{V}^\oplus$ and complete the proof.

\subsection{Proof of Theorem~\ref{leastthm}}  \label{PMPOMapdA6}
  We give the proof based on the optimal~$(0,1)$-matrix completion formulations, namely, Problem~(\ref{basicmajmatrix1}) and Problem~(\ref{basicmajmatrix2}).

  First of all, we shall prove a conjecture that there always exists an element $\bm w\in \mathcal{V}^\ominus$ such that $\bm w \prec \bm u$ and $\bm w \prec \bm v$ for an arbitrary pair of elements $\bm u, \bm v\in \mathcal{V}^\ominus$. We shall show this result by induction on $m$, which is the number of rows of a~$(0,1)$-matrix in Problem~(\ref{basicmajmatrix1}).

  If $m=1$, let $\bm a$ be $[1~\cdots~1~0~\cdots~0]'\in \mathcal{{X}}^\ominus$. 
  It follows from Lemma~\ref{deducadditionseq} that~$\bm w=\bm c- \bm a$ satisfies $\bm w \prec \bm u$ for every~$\bm u\in \mathcal{V}^\ominus$. By the hypothesis that the conjecture holds for $m\leq k-1$, we will show that such a vector $\bm w$ also exists when $m=k$.

  Let $P$ and $Q$ be two distinct feasible $(0,1)$-matrices in Problem~(\ref{basicmajmatrix1}) with the same size~$k\times n$ and they respectively give~$\bm u = \bm c-\sum_{i=1}^{k}\bm p_{i\bigcdot}'$ and $\bm v = \bm c-\sum_{i=1}^{k}\bm q_{i\bigcdot}'$. By hypothesis, there is a vector~$\bm{\hat{w}}$ and a $(k-1)\!\times n$ $(0,1)$-matrix~$A$ such that
  \begin{align*}
     & \sum\nolimits_{j=1}^{n}a_{ij}=r_i \text{, for } 1\leq i \leq k-1; \;\bm{\hat{w}}\prec \bm c-\sum\nolimits_{i=1}^{k-1}\bm p_{i\bigcdot}';\\
     & \bm{\hat{w}}\prec \bm c-\sum\nolimits_{i=1}^{k-1}\bm q_{i\bigcdot}';\;\bm{\hat{w}}=\bm c-\sum\nolimits_{i=1}^{k-1}\bm a_{i\bigcdot}'.
  \end{align*}
  We can obtain~$\bm u,\bm v\in\mathbb{N}^n$ by deducting ones from~$r_k$ distinct elements of a certain nonnegative vector which is no smaller than $\bm{\hat{w}}$ in majorization, so there are at least $r_k$ elements larger than zero in $\bm{\hat{w}}$. Specify the positions of the~$r_k$ largest elements in~$\bm{\hat{w}}$. In the case of ties, pick them randomly such that exactly~$r_k$ positions are identified. Then, define a $k\times n$ $(0,1)$-matrix~$\bar{A}$ as follows. The first $k-1$ rows of~$\bar{A}$ are given by~$A$, while there are $r_k$ ones in the last row of~$\bar{A}$, appearing at the~$r_k$ positions just identified. It follows from Lemma~\ref{deducadditionseq} that~\mbox{$\bm w \prec \bm u$} and~$\bm w \prec \bm v$, where~$\bm{{w}}=\bm c-\sum_{i=1}^{k}\bm{\bar{a}}_{i\bigcdot}'$. At this point, we have shown the proposed conjecture is true. Moreover, because~$\mathcal{V}^\ominus$ has finite elements, we conclude that all the minimal elements of $\mathcal{V}^\ominus$ lie in the same equivalent class. By Proposition~\ref{propcan}, we show that the canonical element of the equivalent class also belongs to~$\mathcal{V}^\ominus$. Finally, by definition, these minimal elements are also least elements in the preordered set. 

  So far, we have proven the case regarding~$\mathcal{V}^\ominus$, while the case regarding~$\mathcal{V}^\oplus$ can be proven similarly, by Proposition~\ref{propcan} and Lemma~\ref{deducadditionseq}. Thus, we complete the proof.

\subsection{Proof of Proposition~\ref{feacond}} \label{PMPOMapdA2}
\vspace{-1em}
\begin{algorithm}[h!]
\caption{Obtain~$\bm {x}\in \mathbb{N}^n$ with~$\bm {x} \prec \bm t$ and $\bm {x} \leq \bm c$}
\label{GethAlg}
    \KwIn {Two vectors $\bm c\in \mathbb{N}^n$ and $\bm t\in \mathbb{N}^n$ with $\bm c \prec^w \bm t$.}
    \KwOut  {A vector~$\bm {x}\in \mathbb{N}^n$ where $\bm {x} \prec \bm t$ and $\bm {x} \leq \bm c$.}
    {{\bf Initialization:} $\bm {x}=\bm c, k=n$\;}
    \Repeat{$k=0$}{
    { Decrease ${x}_k$ such that one of the inequalities $\left(\sum_{i=j}^{n}{x}_{i}\geq \sum_{i=j}^{n}t_{i}, j\in \underline{k}\right)$ is tight:}\newline 
    { If $\hat{p}$ is the smallest index such} {that~$\sum_{i=\hat{p}}^{n}{x}_{i}= \sum_{i=\hat{p}}^{n}t_{i}$, then $k=\hat{p}-1$\;}
    }
\end{algorithm}
\vspace{-1em}
The necessity follows from Lemma~\ref{realfeacondcombineprecminusmaj}. Set~$\bm t=\bm r^*$. To prove the sufficiency, we design Algorithm~\ref{GethAlg} to find an integer vector~$\bm {x}$ with~$\bm {x}\leq \bm c$ and~$\bm {x}\prec \bm t$. We prove its correctness as follows. First, we have ${x}_n\geq t_n \geq 0$ and~$\bm {x}\leq \bm c$ when implementing Algorithm~\ref{GethAlg}. Second, since~$\bm c \prec^w \bm t$, we can always find an index~$\hat{p}$ in each iteration of Algorithm~\ref{GethAlg}. Moreover, for every new~$k$, we see that $$\sum\limits_{i=k}^{n}{x}_{i}\geq \sum\limits_{i=k}^{n}t_{i},\;\sum\limits_{i=k+1}^{n}{x}_{i}= \sum\limits_{i=k+1}^{n}t_{i}\text{ and }\sum\limits_{i=k+2}^{n}{x}_{i}\geq \sum\limits_{i=k+2}^{n}t_{i}.$$
It follows that ${x}_k\geq t_k\geq t_{k+1}\geq {x}_{k+1}$. Therefore, the vector~$\bm {x}$ obtained by Algorithm~\ref{GethAlg} is in its nonincreasing form. {Then, we conclude by Definition~\ref{defmaj} that~$\bm {x}\prec \bm t$}. By simple algebra, we show that~$\bm c \prec^w \bm r^*$ if and only if~$\bm r \prec_w \bm c^*$, which completes the proof.

\subsection{Proof of Proposition~\ref{rescanon}} \label{PMPOMapdA3}
If $\bm {x} \in \mathcal{{X}}^\ominus$, then $\bm {x}\prec \bm t$ and $\bm {x}\leq \bm c=\bm c^\downarrow$. By Lemma~\ref{comadjswap}, we have $\bm {x}^{\downarrow} \prec \bm t$ and $\bm {x}^{\downarrow} \leq \bm c$, which implies~$\bm {x}^{\downarrow} \in \mathcal{{X}}^\ominus$. By Lemma~\ref{comadjswap2}, we prove~$\bm c -\bm {x}^{\downarrow} \prec \bm c - \bm {x}$. Next, if $\bm {x} \in \mathcal{{X}}^\oplus$, then~$\bm {x}^{\uparrow} \prec \bm t$. According to Lemma~\ref{comadjswap2} and~$\bm b=\bm b^\downarrow$, we further have $\bm b + \bm {x}^{\uparrow}\prec \bm b + \bm {x}$, which completes the proof.

\subsection{Proof of Proposition~\ref{feaslattice}} \label{PMPOMapdAf}
The part with~$\mathcal{{X}}^\oplus_\downarrow$ follows from the transitivity of majorization and the definition of sublattice. Next, we prove the part with~$\mathcal{{X}}^\ominus_\downarrow$. Set~$\bm t=\bm r^*$. Given~$\bm x,\bm y\in \mathcal{{X}}^\ominus_\downarrow$, it follows from the definition that~$\bm x\prec \bm t$, $\bm y\prec \bm t$, $\bm x\leq \bm c$, and~$\bm y\leq \bm c$. For notational convenience, we denote~$\inf_{\mathcal{N}_{\|\bm t\|_1}}\{\bm x, \bm y\}$ by $\bm z$. It follows from the transitivity that~$\bm z\prec \bm t$. In what follows, we shall show that~$\bm z \leq \bm c$. First, for all~$k\in \underline{n}$, we have
  $$z_k=\min\left\{\sum\limits_{i=1}^{k} x_{[i]},\sum\limits_{i=1}^{k} y_{[i]}\right\}-\min\left\{\sum\limits_{i=1}^{k-1} x_{[i]},\sum\limits_{i=1}^{k-1} y_{[i]}\right\}.$$
  If~$\sum_{i=1}^{k-1} x_{[i]}=\min\left\{\sum_{i=1}^{k-1} x_{[i]},\sum_{i=1}^{k-1} y_{[i]}\right\}$, then we have \begin{align*}z_k&=\min\left \{\sum\limits_{i=1}^{k} x_{[i]},\sum\limits_{i=1}^{k} y_{[i]}\right\}-\sum\limits_{i=1}^{k-1} x_{[i]}\\
  &\leq  \sum\limits_{i=1}^{k} x_{[i]}-\sum\limits_{i=1}^{k-1} x_{[i]}= x_{[k]}.
    \end{align*}
    Using a similar argument, for the other case, we can verify that~$z_k\leq y_{[k]}$. Thus, we show that~$z_k\leq \max\{x_{[k]},y_{[k]}\}\leq c_k$. It follows that~$\bm z$ belongs to~$\mathcal{{X}}^\ominus_\downarrow$.

  It remains to show that~$\sup_{\mathcal{{X}}^\ominus_\downarrow}\{\bm x, \bm y\}=\sup_{\mathcal{N}_{\|\bm t\|_1}}\{\bm x, \bm y\}$. We denote~$\sup_{\mathcal{N}_{\|\bm t\|_1}}\!\{\bm x, \bm y\}$ by $\bm z$. Since~$\bm x\prec \bm t$, and $\bm y \prec \bm t$, we have $\bm z\prec \bm t$. Next, let us show~$\bm z\leq \bm c$. We can obtain the elements of~$\bm z$ by the recursive formula below: for each~$k\in \underline{n}$,
  \begin{equation*}
    \begin{split}
        ~~~~~~~~~&z_k=\min_{\alpha\in \mathbb{N} }~~\alpha\ \\
        \text{subject to~~}&  \textstyle \sum\limits_{i=1}^{k-1}z_{i}+j\alpha \geq \sum\limits_{i=1}^{k-1+j}x_{[i]},\\
        &  \textstyle \sum\limits_{i=1}^{k-1}z_{i}+j\alpha \geq \sum\limits_{i=1}^{k-1+j}y_{[i]}, \text{for } 1\leq j \leq n-k+1.
    \end{split}
  \end{equation*}
   As a result, there exists a special index $p$ such that
   \begin{align*}&1\leq p\leq n-k+1 \text{ and }\\
   &\sum\limits_{i=1}^{k-1}z_{i}+p(z_k-1) < \max\left\{\sum\limits_{i=1}^{k-1+p}x_{[i]},\sum\limits_{i=1}^{k-1+p}y_{[i]}\right\}.
    \end{align*}
    Moreover, we have~$\sum_{i=1}^{k-1}z_{i}\geq \max\left\{\sum_{i=1}^{k-1}x_{[i]},\sum_{i=1}^{k-1}y_{[i]}\right\}$. Therefore, it follows that $$p(z_k-1) <\max\left\{\sum\limits_{i=k}^{k-1+p}x_{[i]},\sum\limits_{i=k}^{k-1+p}y_{[i]}\right\}.$$
It follows from $\bm x, \bm y\in \mathbb{R}^n_{\downarrow}$ that $z_k\leq \max\{x_{[k]},y_{[k]}\}\leq c_k$. Hence, $\bm z$ belongs to~$\mathcal{{X}}^\ominus_\downarrow$. 
Thus, we complete the proof.

\subsection{$(\mathcal{V}^\ominus_\downarrow,\prec)$ and $(\mathcal{V}^\oplus_\downarrow,\prec)$ Are in General Not Lattices} \label{PMPOMapdA5}
  It suffices to prove the stated result by presenting two counterexamples, where $(\mathcal{V}^\ominus_\downarrow,\prec)$ and $(\mathcal{V}^\oplus_\downarrow,\prec)$ are not lattices.
\begin{example}
  Consider a setup with~$\bm c=\begin{bmatrix}8~6~6~6~4~4~4\end{bmatrix}'$ and~$\bm t=\begin{bmatrix}2~2~1~1~0~0~0\end{bmatrix}'$. We first set~$\bm {x} \in \mathcal{{X}}^\ominus$ in order as
  \begin{equation*}
  \begin{split}
     & \begin{bmatrix}1~0~1~2~0~0~2\end{bmatrix}', \; \begin{bmatrix}2~0~0~1~0~1~2\end{bmatrix}', \\
     & \begin{bmatrix}2~0~0~2~0~1~1\end{bmatrix}',  \; \begin{bmatrix}2~0~1~1~0~0~2\end{bmatrix}'.
     \end{split}
  \end{equation*}
Then, we obtain four distinct elements~$\left(\bm u, \bm v, \bm w,\bm z\right)$ in $\mathcal{V}^\ominus_\downarrow$, which are stated below:
  \begin{equation*}
  \begin{split}
     & \begin{bmatrix}7~6~5~4~4~4~2\end{bmatrix}',  \; \begin{bmatrix}6~6~6~5~4~3~2\end{bmatrix}', \\
     & \begin{bmatrix}6~6~6~4~4~3~3\end{bmatrix}',  \; \begin{bmatrix}6~6~5~5~4~4~2\end{bmatrix}'.
     \end{split}
  \end{equation*}

  For $\bm x, \bm y$ in~$(\mathcal{N}_\tau,\prec)$, we say $\bm y$ covers $\bm x$ if $\bm x\neq \bm y$, $\bm x \prec \bm y$, and there exists no third element~$\bm z \in \mathcal{N}_\tau$ such that~\mbox{$\bm x \prec \bm z \prec \bm y$}. It follows from~\cite[Section 5.D]{olkin2016inequalities} that~$\bm x$ is covered by~$\bm y$ if and only if~$\bm x=\bm y -\bm {e_i} + \bm {e_j}$ where $y_i>y_j+1$. Therefore, in~$\left({\mathcal{N}_{\|\bm c\|_1-\|\bm t\|_1}},\prec\right)$, the two elements, $\bm{u}$ and $\bm v$, cover a common vector~$[6~6~6~4~4~4~2]'$, which also covers~$\bm w$ and~$\bm z$ both. However, the common vector does not belong to the canonical attainable set~$\mathcal{V}^\ominus_\downarrow$, which implies $\mathcal{V}^\ominus_\downarrow$ is not a lattice under majorization. 
\end{example}

\begin{example}
  Consider a setup with~$\bm b=\begin{bmatrix}4~4~4~2~2~2~0\end{bmatrix}'$ and~$\bm t=\begin{bmatrix}2~2~1~1~0~0~0\end{bmatrix}'$. We first set $\bm {x} \in \mathcal{{X}}^\oplus$ in order as
  \begin{equation*}
  \begin{split}
     & \begin{bmatrix}2~0~0~2~1~0~1\end{bmatrix}', \; \begin{bmatrix}2~1~0~1~0~0~2\end{bmatrix}', \\
     & \begin{bmatrix}2~0~0~1~1~0~2\end{bmatrix}',\;  \begin{bmatrix}1~1~0~2~0~0~2\end{bmatrix}'.
     \end{split}
  \end{equation*}
  Then, we obtain four distinct elements~$\{\bm u, \bm v, \bm w,\bm z\}$ in $\mathcal{V}^\oplus_\downarrow$, which are stated as follows:
  \begin{equation*}
  \begin{split}
     & \begin{bmatrix}6~4~4~4~3~2~1\end{bmatrix}', \; \begin{bmatrix}6~5~4~3~2~2~2\end{bmatrix}', \\
     & \begin{bmatrix}6~4~4~3~3~2~2\end{bmatrix}', \; \begin{bmatrix}5~5~4~4~2~2~2\end{bmatrix}'.
     \end{split}
  \end{equation*}
   In~$\left({\mathcal{N}_{\|\bm c\|_1+\|\bm t\|_1}},\prec\right)$, the elements~$\bm{u}$ and $\bm v$ cover a common vector~$[6~4~4~4~2~2~2]'$, which also covers~$\bm w$ and~$\bm z$ both. However, the common vector is not in the canonical attainable set~$\mathcal{V}^\oplus_\downarrow$.
Thus, we show $\mathcal{V}^\oplus_\downarrow$ is not a majorization lattice.
  \end{example}

\subsection{Proof of Theorem~\ref{algthm}} \label{PMPOMapdA7}
  Given~$\bm c\in \mathbb{N}^n$ and~$\bm r\in \mathbb{N}^m$, Algorithm~\ref{MAlgo1} has a loop of~$m$ iterations and the complexity of each iteration is~$\mathcal{O}(n)$. Thus, the complexity of Algorithm~\ref{MAlgo1} is~$\mathcal{O}(mn)$. Similarly, for~$\bm b \in \mathbb{N}$ and~$\bm r \in \mathbb{N}^m$, the complexity of Algorithm~\ref{MAlgo2} is~$\mathcal{O}(mn)$.

  We herein prove the correctness of Algorithm~\ref{MAlgo1}, while the part on Algorithm~\ref{MAlgo2} can be proven similarly. We do this by hypothesis induction in terms of~$m$ i.e.,~the number of rows. When~$m=1$, the correctness follows from Lemma~\ref{deducadditionseq}. By the hypothesis that the algorithm is correct for $m\leq k$, we move to~$m=k+1$. First, after the $k$th iteration of the second loop, we attain a minimal element $\bm{\bar{c}}^{(k)}$ of the following set:~$$\left\{\bm c-\bm{\hat{{x}}} \mid \bm{\hat{{x}}}\in \mathbb{N}^n, \bm{\hat{{x}}}\leq \bm c \text{, and } \bm{\hat{{x}}}\prec \left(\begin{bmatrix}r_1&r_2&\cdots&r_k\end{bmatrix}'\right)^{*}\right\}.$$

   According to Proposition~\ref{leastthm}, the element $\bm{\bar{c}}^{(k)}$ is equivalent to all other minimal elements of the above set. Note that the threshold vector $\bm t$ can be obtained by adding ones to the $r_{k+1}$ largest elements of $[r_1~r_2~\cdots~r_k]^{'*}$. In particular, we choose elements with the smaller indices in the case of ties such that exactly $r_{k+1}$ elements are increased by one. Then, we can conclude by Lemma~\ref{deducadditionseq} that~$\bm c-\bm{\bar{c}}^{(k+1)}\prec \bm t$.

  It follows from the equivalence between Problem~(\ref{basicmaj1}) and Problem~(\ref{basicmajmatrix1}) that each~$\bm {x}\in \mathcal{{X}}^\ominus$ can be decomposed into~$\bm{\hat{{x}}}+\bm y$, where $\bm{\hat{{x}}}\in \mathbb{N}^n$, $\bm c \geq \bm{\hat{{x}}}$, $\bm{\hat{{x}}}\prec [r_1~r_2~\cdots~r_k]'^{*}$, and~$\bm y$ is a~$(0,1)$-vector with exactly $r_{k+1}$ ones. Define~$\bm{\hat{y}}\in \mathbb{N}^n$ as a~$(0,1)$-vector with exactly $r_{k+1}$ ones appearing at positions of the~$r_{k+1}$ largest elements in~$\bm c-\bm{\hat{{x}}}$. Thus, we have~\mbox{$\bm{\hat{{x}}}+\bm{\hat{y}}\in \mathcal{{X}}^\ominus$}. Furthermore, by Lemma~\ref{deducadditionseq}, we conclude that~$\bm c-\bm{\hat{{x}}}-\bm{\hat{y}}\prec \bm c-\bm{\hat{{x}}}-\bm{y}$. To summarize, for every feasible matrix~$A$ in Problem~(\ref{basicmajmatrix1}), we can find another feasible matrix~$B$ whose first $k$ rows are the same as those of~$A$ and whose ones in the~$(k+1)$th row exactly appear in the positions corresponding to the $r_{k+1}$ largest elements of~$\bm c -\sum_{i=1}^{k}\bm a_{i\bigcdot}'$. Thus, the objective value of Problem~(\ref{basicmajmatrix1}) generated from~$A$ is no smaller than that generated from~$B$ under majorization.

  Based on the hypothesis for the case $m=k$ and Lemma~\ref{deducadditionseq}, we conclude that Algorithm~\ref{MAlgo1} gives a minimal element in~$\mathcal{V}^\ominus$ for $m=k+1$, which completes the proof.

\subsection{Proof of Theorem~\ref{alloptobj}}\label{PMPOMapdAall}
It suffices to focus on Algorithm~\ref{MAlgo1} and the part regarding Algorithm~\ref{MAlgo2} follows from similar arguments. Specifically, we prove this theorem by transforming an optimal feasible solution~$A$ into another~$\hat{A}$ that can be obtained by Algorithm~\ref{MAlgo1} and leads to the same optimal objective value~$\bm v$. If~$A$ cannot be generated by Algorithm~\ref{MAlgo1}, we can find the smallest possible index~$i\in \underline{m}$ and two distinct indices, $p,q\in\underline{n}$ such that
\begin{align*}
   a_{ip} = 0, a_{iq}=1,\text{ and } c_p-\sum\nolimits_{k=1}^{i-1} a_{kp}>  c_q-\sum\nolimits_{k=1}^{i-1} a_{kq}.
\end{align*}
By the optimality of~$A$, we conclude that~$v_p\leq v_q$; otherwise, we can reset~$a_{ip} = 1$ and~$a_{iq}=0$ to obtain another attainable objective value~$\hat{\bm v}$ satisfying~$\hat{\bm v} \prec \bm v$ and~$\hat{\bm v} \nsim \bm v$.

It follows that there exists an index~$j\in \underline{m}$ and~$j>i$ such that~$a_{jp} = 1$ and~$a_{jq}=0$. Then, we update the matrix~$A$ by performing the following interchange on the associated~$2\times 2$ submatrix and keeping other elements unchanged.
$$
\begin{tabular}{c|ccc}
     $\quad$&$p$&$\cdots$&$q$  \\\hline
     $i$&$0$&$\cdots$&$1$\\
     $\vdots$&$\vdots$&$\ddots$&$\vdots$\\
     $j$&$1$&$\cdots$&$0$
\end{tabular}
\Rightarrow
\begin{tabular}{c|ccc}
     $\quad$&$p$&$\cdots$&$q$  \\\hline
     $i$&$1$&$\cdots$&$0$\\
     $\vdots$&$\vdots$&$\ddots$&$\vdots$\\
     $j$&$0$&$\cdots$&$1$
\end{tabular}
$$
The new matrix is also an optimal solution and gives the same objective value. Repeat the above process and we finally obtain an admissible optimal matrix~$\hat{A}$, which completes the proof.

\subsection{Proof of Proposition~\ref{reorder}} \label{PMPOMapdA8}
  We will prove the part with regard to Algorithm~\ref{MAlgo1}, while the claim regarding Algorithm~\ref{MAlgo2} can be proven similarly.

  First, consider a simple case where~$\sigma$ corresponds to an adjacent swap defined in~Appendix~\ref{PMPOMapdA0}. Without loss of generality, we assume that $\sigma_{{p}}={p}+1$, $\sigma_{{p}+1}={p},$ and~$\sigma_{j}=j$, for~\mbox{$j=1,2,\dots,{p}-1,{p}+2,\dots,m$}. Without the permutation~$\sigma$, in the~$p$th iteration of the second loop, the largest~$r_p$ elements of~$\bar{\bm c}^{{p}-1}$ are all decreased by one. In the next iteration, the largest~$r_{p+1}$ elements of the newly obtained vector are reduced by one, which leads to~$\bar{\bm c}^{o1}$. However, with the permutation~$\sigma$, we should exchange~$r_{p}$ and~$r_{p+1}$, which leads to~$\bar{\bm c}^{o2}$. By carefully analyzing the possible ties during the two processes, we observe that $\bar{\bm c}^{o1}\sim \bar{\bm c}^{o2}$. In other words, the rearrangement of~$\bar{\bm c}^{{p}+1}$ remains unchanged after the adjacent swap. Thus, the optimal objective value generated by Algorithm~\ref{MAlgo1} remains equivalent as well.

  Recall that each general permutation can be written as a composition of a sequence of adjacent swaps. Repeatedly applying the above analysis to each involved adjacent-swap permutation, we finally verify this proposition.
\bibliographystyle{IEEEtran}


%
%
%

%

\newpage
\begin{IEEEbiography}[{\includegraphics[width=1in,height=1.25in,clip,keepaspectratio]{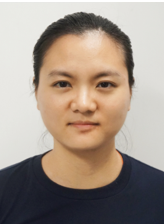}}]{Yanfang Mo}
received the B.E. degree in automation from Zhejiang University, Zhejiang, China, in~2014, and the Ph.D. degree in electronic and computer engineering from the Hong Kong University of Science and Technology, Hong Kong, China, in~2020. From February~2017 to July~2017, she was a Visiting Researcher at the University of California, Berkeley. She is currently a Research Associate in the School of Data Science, City University of Hong Kong.

Her research interests include smart grid, intelligent transportation, construction robotics, process monitoring and fault diagnosis, electricity product design, optimal resource allocation, online algorithm design, and partial order programming. She is the recipient of the best poster paper award at IEEE INFOCOM~2021 and a Hong Kong PhD Fellowship Scheme awardee.
\end{IEEEbiography}

\vskip 0pt plus -1fil
\begin{IEEEbiography}[{\includegraphics[width=1in,height=1.25in,clip,keepaspectratio]{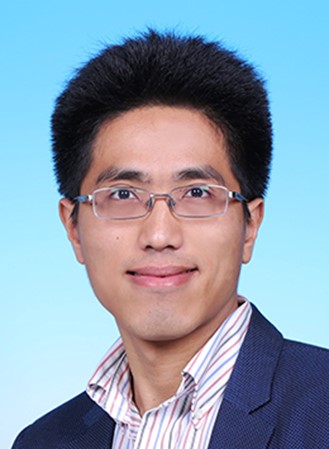}}]{Wei Chen}
received the B.S. degree in engineering and the double B.S. degree in economics from Peking University, Beijing, China, in 2008. He received the M.Phil. and Ph.D. degrees in electronic and computer engineering from the Hong Kong University of Science and Technology, Hong Kong S.A.R., China, in 2010 and 2014, respectively. He is currently an Assistant Professor in the Department of Mechanics and Engineering Science at Peking University. Prior to joining Peking University, he worked in the ACCESS Linnaeus Center of KTH Royal Institute of Technology and the EECS Department of University of California at Berkeley for postdoctoral research, and in the ECE Department of the Hong Kong University of Science and Technology as a Research Assistant Professor.

His research interests include linear systems and control, networked control systems, optimal control, smart grid, cyber physical security, and network science. He was the recipient of the best student paper award at the 2012 IEEE International Conference on Information and Automation.
\end{IEEEbiography}

\vskip 0pt plus -1fil
\begin{IEEEbiography}[{\includegraphics[width=1in,height=1.25in,clip,keepaspectratio]{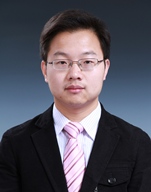}}]{Keyou You}
(SM’$17$) received the B.S. degree in Statistical Science from Sun Yat-sen University, Guangzhou, China, in~$2007$ and the Ph.D. degree in Electrical and Electronic Engineering from Nanyang Technological University (NTU), Singapore, in~$2012$. After briefly working as a Research Fellow at NTU, he joined Tsinghua University in Beijing, China where he is now a tenured Associate Professor in the Department of Automation. He held visiting positions at Politecnico di Torino, the Hong Kong University of Science and Technology, University of Melbourne and etc.

His current research interests include networked control systems, distributed optimization and learning, and their applications.
Dr. You received the Guan Zhaozhi award at the~$29$th Chinese Control Conference in~$2010$ and the ACA (Asian Control Association) Temasek Young Educator Award in~$2019$. He received the National Science Fund for Excellent Young Scholars in~$2017$. He is serving as an Associate Editor for the IEEE Transactions on Cybernetics, IEEE Transactions on Control of Network Systems, IEEE Control Systems Letters (L-CSS), Systems~$\&$ Control Letters.
\end{IEEEbiography}

\vskip 0pt plus -1fil
\begin{IEEEbiography}[{\includegraphics[width=1in,height=1.25in,clip,keepaspectratio]{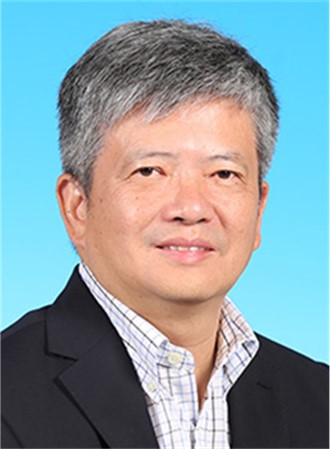}}]{Li Qiu}
(F’$07$) received his Ph.D. degree in electrical engineering from the
University of Toronto in 1990. After briefly working in the Canadian Space Agency, the Fields Institute for Research in Mathematical Sciences (Waterloo),
and the Institute of Mathematics and its Applications (Minneapolis), he joined the Hong Kong University of Science and Technology in 1993, where he is now a
Professor of Electronic and Computer Engineering.

Prof.\ Qiu's research interests include system, control,
optimization theory, and mathematics for information technology, as well as their
applications in manufacturing industry and energy systems.
He is also interested in control education and coauthored an undergraduate textbook
``Introduction to Feedback Control'' which was published by Prentice-Hall in 2009.
He served as an
associate editor of the {\em IEEE Transactions on Automatic Control} and an associate editor of {\em Automatica}. He was the general chair of the 7th Asian Control Conference, which was held in Hong Kong in 2009. He was a Distinguished Lecturer from 2007 to 2010 and was a member of the Board of Governors in 2012 and 2017 of the IEEE Control Systems Society. He is the founding chairperson of the Hong Kong Automatic Control Association and a vice president of Asian Control Association. He is a Fellow of IEEE and a Fellow of IFAC.
\end{IEEEbiography}








\end{document}